\documentclass[11pt,a4paper]{article}
\usepackage[margin=1.25in]{geometry}
\usepackage{amsmath,amsfonts,amssymb,amsthm}
\usepackage{graphicx,color,enumerate}
\usepackage[font=small, format=hang, labelfont=bf]{caption}
\usepackage{subfig}
\usepackage[pdfpagelabels,colorlinks,allcolors=blue]{hyperref}

\newcommand{\arc}[1]{\smash{\stackrel{\frown}{#1}}}

\theoremstyle{plain}
\newtheorem{theorem}{Theorem}
\newtheorem{definition}{Definition}
\newtheorem{remark}{Remark}
\newtheorem{corollary}{Corollary}
\newtheorem{lemma}{Lemma}

\title{On a Conjecture of Lov\'asz on Circle-Representations \newline of Simple 4-Regular Planar Graphs\footnote{This draft is an update of~\cite{DBLP:journals/jocg/BekosR15}, where we fix a typo in the statement of Lemma~\ref{lem:gadget-realization-2}.}}

\author{Michael A.\ Bekos$^{1}$, Chrysanthi N.\ Raftopoulou$^{2}$
\\[0.1in]
$^1$\small Institute for Informatics, University of T{\"u}bingen, T{\"u}bingen, Germany\\
\texttt{\small bekos@informatik.uni-tuebingen.de}
\and
$^2$\small School of Applied Mathematical \& Physical Sciences,\\\small  National Technical University of Athens, Greece\\
\texttt{\small crisraft@mail.ntua.gr}
}
\date{}
          
\begin{document}

\maketitle

\begin{abstract}
Lov\'asz conjectured that every connected 4-regular planar graph $G$
admits a \emph{realization as a system of circles}, i.e., it can be
drawn on the plane utilizing a set of circles, such that the
vertices of $G$ correspond to the intersection and touching points
of the circles and the edges of $G$ are the arc segments among pairs
of intersection and touching points of the circles. In this paper,
we settle this conjecture. In particular, (a)~we first provide tight
upper and lower bounds on the number of circles needed in a
realization of any simple 4-regular planar graph, (b)~we
affirmatively answer Lov\'asz's conjecture, if $G$ is 3-connected,
and (c)~we demonstrate an infinite class of simple connected
4-regular planar graphs which are not 3-connected (i.e., either
simply connected or biconnected) and do not admit realizations as a
system of circles.
\end{abstract}

\section{Introduction}

All graphs considered in this paper are simple, finite and
undirected. Given a graph $G$, we denote by $V[G]$ and $E[G]$ the
set of vertices and edges of $G$, respectively. If $G$ is regular,
we denote by $d(G)$ its degree.

\begin{definition}
Let $G$ be a connected 4-regular planar graph. We say that $G$
admits a realization as a \emph{system of circles}, if it can be
drawn on the plane using a set of circles such that (see
Figures~\ref{fig:octahedron_1}-\ref{fig:octahedron_3}):
\begin{enumerate}
  \item The vertex set $V[G]$ is given by the intersection and touching points of the
  circles.
  \item The edge set $E[G]$ is defined by all circular arcs between the intersection
  and touching points of the circles.
\end{enumerate}
In the special case where intersection points are not allowed (i.e.,
there are only touching circles), we say that $G$ admits a
realization as a \emph{system of touching circles} (see
Figures~\ref{fig:octahedron_1} and \ref{fig:octahedron_2}).
\end{definition}

\begin{figure}[thb]
  \centering
  \begin{minipage}[b]{.22\textwidth}
    \centering
    \subfloat[\label{fig:fig_octahedron0}{}]
    {\includegraphics[width=\textwidth]{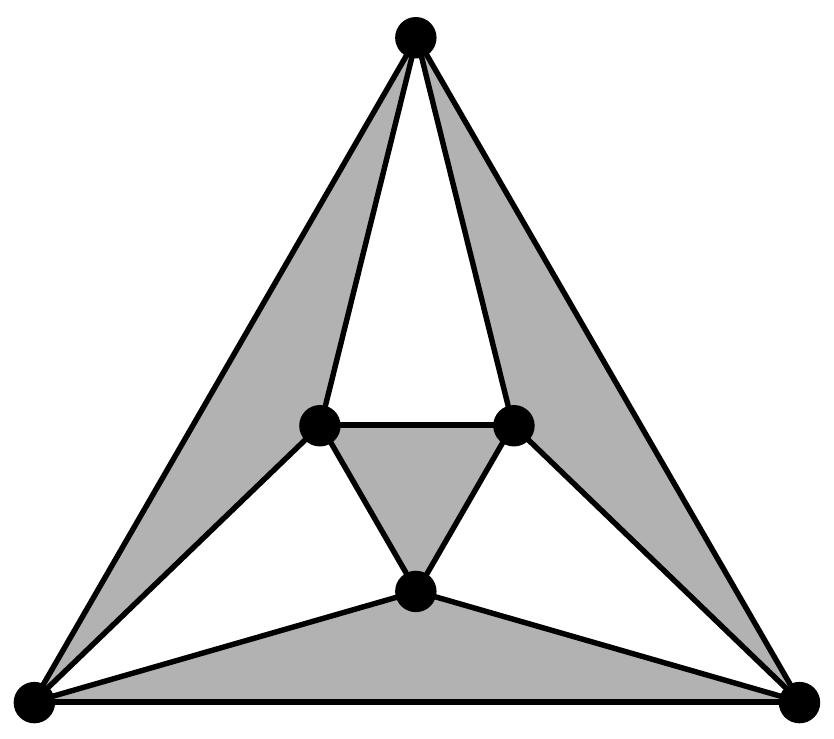}}
  \end{minipage}
  \hfill
  \begin{minipage}[b]{.22\textwidth}
    \centering
    \subfloat[\label{fig:octahedron_1}{}]
    {\includegraphics[width=\textwidth]{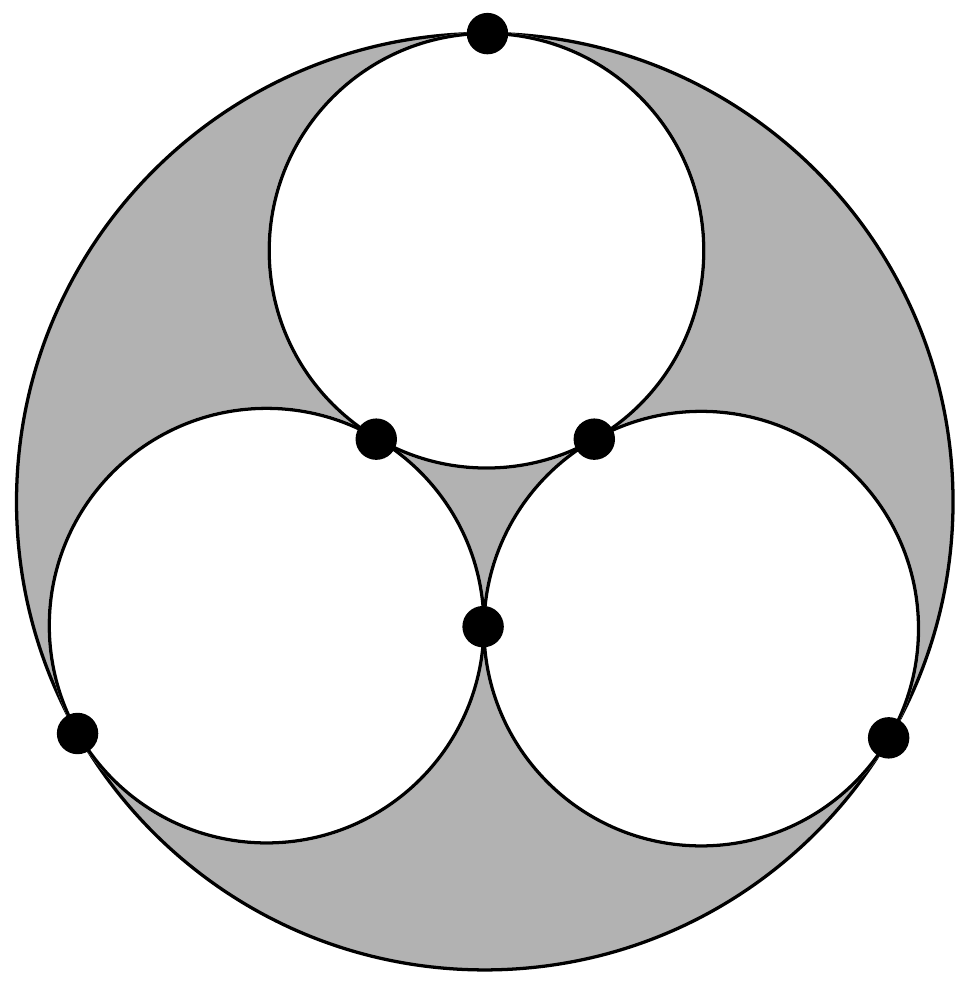}}
  \end{minipage}
  \hfill
  \begin{minipage}[b]{.22\textwidth}
    \centering
    \subfloat[\label{fig:octahedron_2}{}]
    {\includegraphics[width=\textwidth]{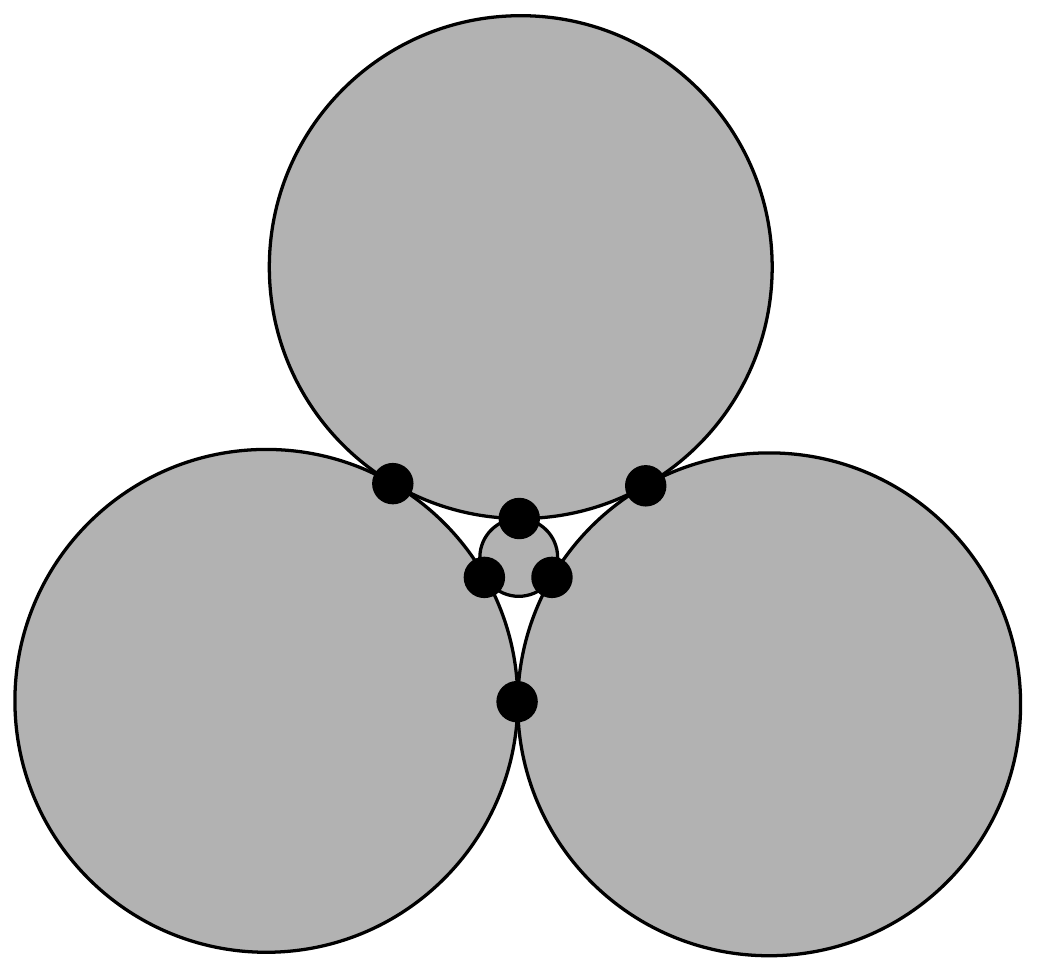}}
  \end{minipage}
  \hfill
  \begin{minipage}[b]{.22\textwidth}
    \centering
    \subfloat[\label{fig:octahedron_3}{}]
    {\includegraphics[width=\textwidth]{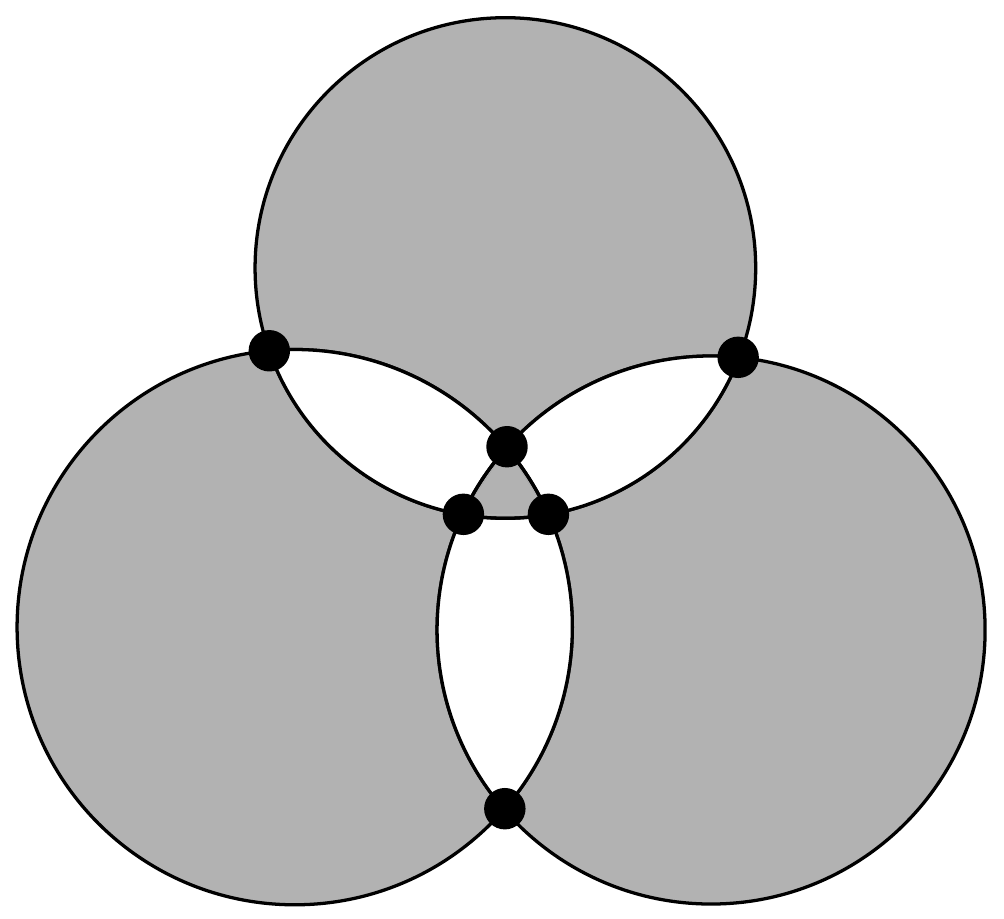}}
  \end{minipage}
  \caption{(a)~A straight-line drawing of the octahedron graph $G_{oct}$.
  (b)-(d)~Non-equivalent realizations of the octahedron graph $G_{oct}$ as system of circles.}
  \label{fig:octahedron}
\end{figure}

Lov\'asz \cite[pp.1175]{EAS},\cite[pp.426]{Le81} conjectured that
every simple connected 4-regular planar graph admits a realization
as a system of circles. To the best of our knowledge this conjecture
remained unanswered. Touching points are necessary, since if we use
only crossings, we have an even number of vertices, but there are
4-regular planar graphs with an odd number of vertices \cite{Ma79}.

In this paper, we prove that every 3-connected 4-regular planar
graph admits a realization as a system of touching circles. If the
input graph is not 3-connected, we demonstrate by an example that a
realization as a system of circles is not always possible.

This paper is structured as follows: Section~\ref{sec:related}
overviews previous work related to this paper. In
Section~\ref{sec:Bounds}, we present bounds on the number of circles
needed in a realization of a simple connected 4-regular planar graph
as a system of circles. In Section~\ref{sec:3-connected}, we prove
that every 3-connected 4-regular planar graph admits a realization
as a system of touching circles. In Sections~\ref{sec:connected} and
\ref{sec:biconnected}, we prove that there exist infinitely many
(either simply connected or biconnected) graphs that do not admit
realizations as system of circles. We conclude in
Section~\ref{sec:conclusions} with open problems and future work.

\section{Related Work}
\label{sec:related}

Closely related to the problem we study is the contact graph
representation problem. A \emph{contact graph} is a graph whose
vertices are represented by geometric objects and whose edges
correspond to two objects touching in some specific predefined way.
There is a rich literature related to various types of contact
graphs, dating back to 1936 in Koebe's theorem \cite{Ko36} which
states that any planar graph can be represented as a contact graph
of disks in the plane. Typical classes of such objects are curves,
line segments, disks, triangles and rectangles
(cf.~\cite{Hl01,Hl98}). Note that Koebe's theorem, which is also
known as \emph{circle-packing theorem}, is the starting point for
the proof of Section~\ref{sec:3-connected}. The main difference
between the problem of representing a graph as a contact graph of
disks and the problem we study is that in the former problem the
vertices correspond to disks and the edges are implied by the
contacts; in our problem, however, the vertices are the crossing
and/or touching points of the disks and the edges are the
arc-segments defined between them.

\emph{Lombardi drawings}, which attempt to capture some of the
visual aesthetics used by the American artist Mark Lombardi, are
also closely related to our problem. Two features that stand out in
Lombardi's work are the use of circular-arc edges and their even
distribution around vertices. Such even spacing of edges around each
vertex (also known as \emph{perfect angular resolution}) together
with the usage of circular arc edges, formally define Lombardi
drawings of graphs~\cite{DEGKN10,DEGKN13}. Chernobelskiy et
al.~\cite{CCK11} relax the perfect angular resolution constraint in
Lombardi drawings and describe functional force-directed algorithms,
which produce aesthetically appealing near-Lombardi drawings.

Connected 4-regular planar graphs form a well studied class of
graphs. Manca~\cite{Ma79} proposed four operations to generate all
connected 4-regular planar graphs from the octahedron graph.  As
noted by Lehel~\cite{Le81}, Manca's construction could not generate
all connected 4-regular planar graphs, however, an additional
operation could fix this problem. Broersma et.\ al~\cite{BD93}
showed that all 3-connected 4-regular planar graphs can also be
generated from the octahedron, using only three operations.

In the context of graph drawing, 4-regular planar graphs (which can
be viewed as a special case of max-degree 4 planar graphs) have a
long tradition, dating back to VLSI layouts and floor-planning
applications. The main goal in this context is to produce drawings
(referred to as \emph{orthogonal drawings}) in which each vertex
corresponds to a point on the integer grid and each edge is
represented by a sequence of horizontal and vertical line segments.
Pioneering work on orthogonal drawings was done in relation to
VLSI-design by Valiant \cite{Va81}, Leiserson~\cite{Lei80} and
Leighton~\cite{Leg81}. Later on the problem of constructing
orthogonal drawings of maximum degree 4 planar graphs was considered
by Tamassia \cite{Ta87}, Tamassia and Tollis \cite{TT89}, and Biedl
and Kant \cite{BK94}. The objectives here have been the minimization
of the used area, the total edge length, the total number of bends,
the maximum number of bends per edge, and others.

\section{Bounds on the Number of Circles Needed in a Circle Representation of a Simple 4-regular Planar Graph}
\label{sec:Bounds}

In this section, we present upper and lower bounds on the number of
circles needed in a realization of a simple connected 4-regular
planar graph as a system of circles. We also prove that these bounds
are tight, i.e., there exist infinitely many connected 4-regular
planar graphs that admit realizations as system of circles and use
the number of circles given by the two bounds.

\begin{lemma}
Let $G$ be a simple connected 4-regular planar graph on $n$
vertices. Then, the number of circles, say $c[G]$, that participate
in a realization of $G$ as system of circles, if one exists,
satisfies the following inequality: $$(1+\sqrt{1+4n})/2\leq c[G]\leq
2n/3$$ \label{lem:bounds}
\end{lemma}
\begin{proof}
In general, there are certain restrictions concerning the number of
circles that participate in a realization of a graph as a system of
circles:

\begin{itemize}
  \item \emph{Two circles may have at most two vertices in common}:
  Two crossing points if they intersect, one touching point if they
  are tangent, or none if they are separated.
  \item \emph{There exist at least three vertices on every circle},
  since we consider only simple graphs.
  \item \emph{Every vertex belongs to exactly two circles}, since
  every vertex has degree 4.
\end{itemize}

From the latter two properties, it follows that in any realization
of $G$ as a system of circles, the number of vertices of $G$ defined
by all circles is at least $3c[G]/2$, which immediately implies the
desired upper bound, i.e, $c[G]\leq 2n/3$. On the other hand, if
every pair of circles defines exactly two vertices (i.e., every pair
of circles intersects), the corresponding realization of $G$ has the
minimum number of circles. However, in this case a total of
$c[G](c[G]-1)/2$ pairs of circles define at most $c[G](c[G]-1)$
vertices. Hence, it follows that $n \leq c[G](c[G]-1)$ or
equivalently that $(1+\sqrt{1+4n})/2\leq c[G]$.
\end{proof}

In the following, we prove that the bounds of Lemma~\ref{lem:bounds}
are tight.

\begin{lemma}
There exist infinitely many connected 4-regular $n$-vertex planar
graphs that admit realizations as system of $(1+\sqrt{1+4n})/2$ and
$2n/3$ circles. \label{lem:bounds_tight}
\end{lemma}
\begin{proof}
In order to prove this lemma, it is enough to show that there exist
two classes of graphs that admit realizations as system of circles,
in which (a)~every pair of circles intersect and (b)~every circle
has exactly three vertices.

\begin{enumerate}[a.]
\item We aim to create a set of non-coincident circles all with the
same radius and all containing the same interior point. More
formally, let $c \geq 3$ be an arbitrary integer. Let also $C$ be an
auxiliary geometric circle of radius $R>0$ (refer to the dashed
circle of Fig.\ref{fig:lower_bound}). We proceed to draw $c$ circles
of the same radius $r>0$ centered at the vertices of a regular
$c$-gon inscribed at circle $C$, such that $r>R$ (see
Fig.\ref{fig:lower_bound}). It is not difficult to see that every
pair of circles intersects, since their centers are at distance less
than $2r$. One can also appropriately choose $r$, so that no three
circles pass through the same point. Since $c$ is arbitrarily
chosen, the class of graphs derived by the circle representations
corresponding to different values of integer $c$, $c \geq 3$, has
obviously the property that its members admit realizations in which
every pair of circles intersect.
\item In order to prove that there exist infinitely many graphs that admit
realizations as system of circles, in which every circle has exactly
three vertices, we follow a similar approach as in the previous case
(refer to Fig.\ref{fig:upper_bound}). Let $m >2$ be an odd integer
number. We proceed to draw a ``chain of circles'' consisting of $m$
circles of equal radius respectively, which touch with each other
and also touch the interior of an ``enclosing circle'', as
illustrated in Fig.\ref{fig:upper_bound}. The construction ensures
that every circle has exactly three vertices. Hence, the class of
graphs derived by the circle representations corresponding to
different values of $m >2$, has obviously the property that its
members admit realizations in which every circle has exactly three
vertices.
\end{enumerate}
\end{proof}

\begin{figure}[thb]
  \centering
  \begin{minipage}[b]{.4\textwidth}
    \centering
    \subfloat[\label{fig:lower_bound}{}]
    {\includegraphics[width=.9\textwidth]{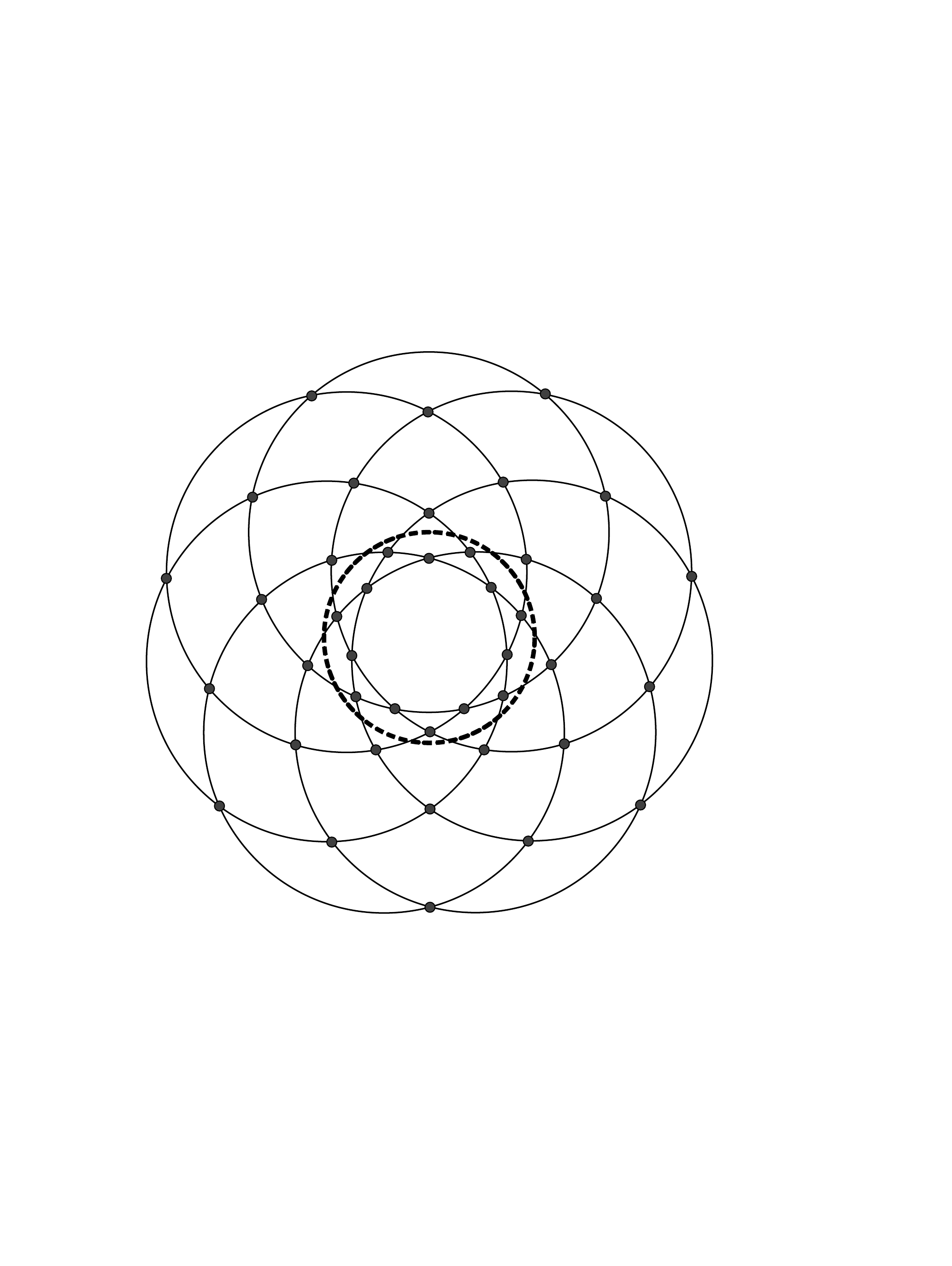}}
  \end{minipage}
  \begin{minipage}[b]{.55\textwidth}
    \centering
    \subfloat[\label{fig:upper_bound}{}]
    {\includegraphics[width=.9\textwidth]{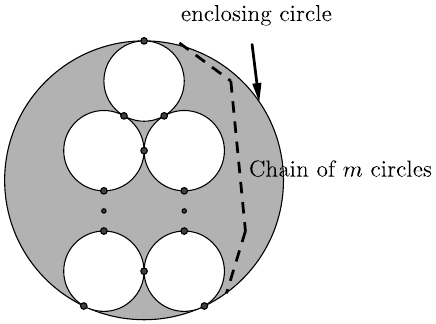}}
  \end{minipage}
  \caption{Circle representations in which (a)~every pair of circles intersects and (b)~every
  circle has exactly three vertices.}
  \label{fig:bounds}
\end{figure}

\section{The case of 3-connected 4-regular planar graphs}
\label{sec:3-connected}

We prove that a 3-connected 4-regular planar graph admits a
realization as a system of touching circles.  Our starting point is
the \emph{circle packing theorem} \cite{Ko36}. A \emph{circle
packing} is a ``connected collection'' of touching circles with
disjoint interiors. The \emph{intersection graph} (also known as
\emph{tangency} or \emph{contact graph}) of a circle packing is the
graph having a vertex for each circle and an edge for every pair of
circles that are tangent. A graph that admits a realization as a
system of touching circles is called \emph{coin graph} (see
Fig.\ref{fig:coin-graph}). Coin graphs are always simple, connected
and planar. The circle packing theorem states the following.

\begin{theorem}[Circle packing theorem~\cite{Ko36}]
For every simple connected planar graph $G$, there is a circle
packing in the plane with $G$ as its intersection graph.
\label{thm:circle-packing}
\end{theorem}

\begin{figure}[htb]
  \centering
  \begin{minipage}[b]{.48\textwidth}
    \centering
    \subfloat[\label{fig:coin_graph_before}{}]
    {\includegraphics[width=.55\textwidth]{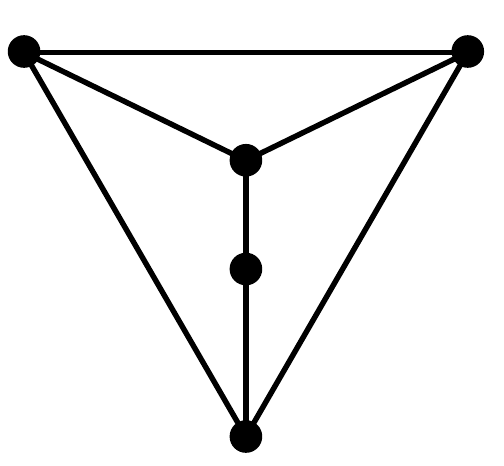}}
  \end{minipage}
  \begin{minipage}[b]{.48\textwidth}
    \centering
    \subfloat[\label{fig:coin_graph_after}{}]
    {\includegraphics[width=.55\textwidth]{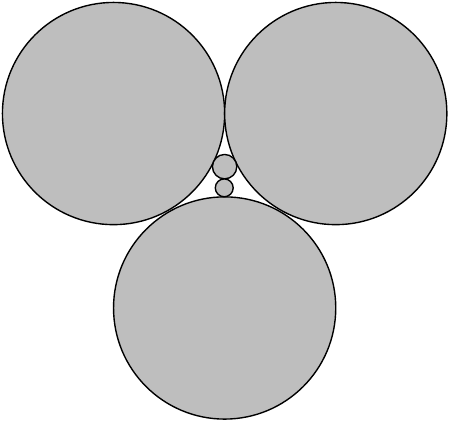}}
  \end{minipage}
  \caption{(a)~A planar graph. (b)~Its representation as a system of touching circles.}
  \label{fig:coin-graph}
\end{figure}

We will use Theorem~\ref{thm:circle-packing} on an auxiliary graph
that can be constructed based on any 3-connected 4-regular planar
graph.

It is well known that a connected planar graph is Eulerian if and
only if its dual is bipartite \cite[pp.172]{Bol98}. Let $G$ be an
embedded 4-regular planar graph. Since $G$ is obviously Eulerian,
its dual $G^*$ is bipartite. Hence, we can color the faces of $G$
using two colors, say gray and white, so that any two adjacent faces
are of different colors. For convenience, we assume that the outer
face of $G$ is always colored white. We proceed to construct a new
graph $IL(G)$ as follows. We associate a vertex of $IL(G)$ with
every gray face of $G$. We join two vertices of $IL(G)$ with an edge
if and only if the corresponding faces of $G$ have at least one
vertex in common (refer to the black colored graph of
Fig.\ref{fig:non3con}).

\begin{figure}[t]
  \centering
  \begin{minipage}[b]{.48\textwidth}
    \centering
    \subfloat[\label{fig:non3con}{}]
    {\includegraphics[width=\textwidth]{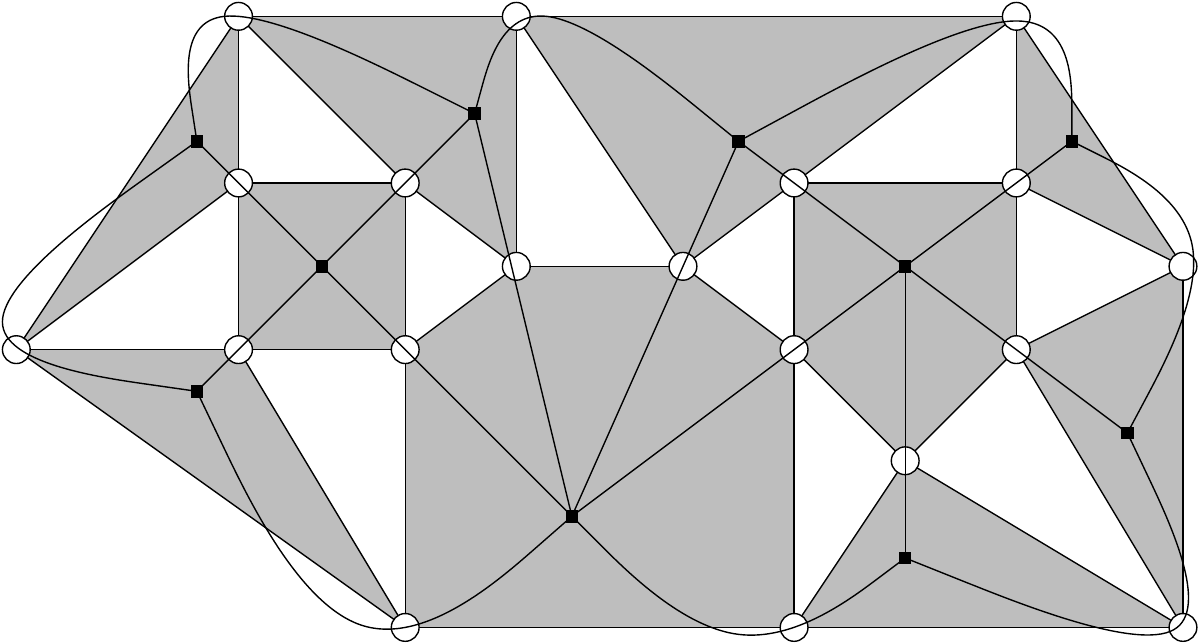}}
  \end{minipage}
  \begin{minipage}[b]{.48\textwidth}
    \centering
    \subfloat[\label{fig:non3connew}{}]
    {\includegraphics[width=\textwidth]{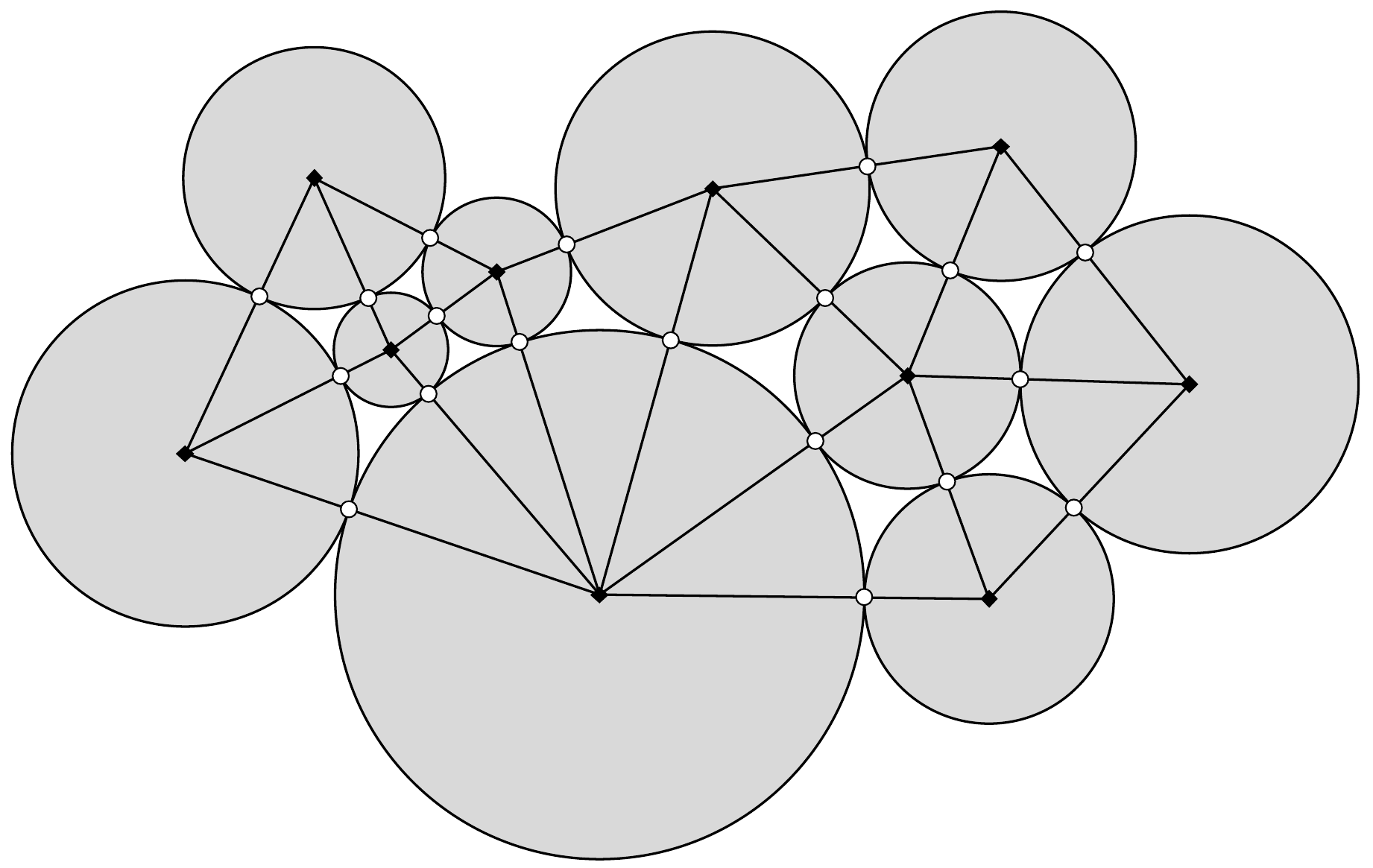}}
  \end{minipage}
  \caption{(a)~Constructing graph $IL(G)$. (b)~A realization of $G$ as system of circles.}
  \label{fig:non3con_case}
\end{figure}

\begin{lemma}
If $G$ is a 3-connected 4-regular planar graph, then $IL(G)$ is
simple. \label{lem:lg_simple}
\end{lemma}
\begin{proof}
Suppose that $G$ is 3-connected and assume for the sake of
contradiction that $IL(G)$ is not simple. W.l.o.g., we further
assume that $IL(G)$ contains a multiple edge, say a double edge
between $f$ and $g$, where $f,g \in V[IL(G)]$ (see
Fig.\ref{fig:l_g_simple}). The case where $IL(G)$ contains selfloops
is treated similarly. By definition, $f$ and $g$ correspond to gray
faces of $G$ that have exactly two common vertices, say $u,v \in
V[G]$. Then, $u$, $f$, $v$ and $g$ define a separating simple closed
curve which intersects $G$ at exactly two vertices (i.e., vertices
$u$ and $v$; see Fig.\ref{fig:l_g_simple}). Note that since $G$ is
simple there is at least one vertex of $G$ that lies in the interior
of this curve and one in its exterior. Hence, $G$ is not
$3$-connected, which implies the claimed contradiction.
\end{proof}

\begin{figure}[t]
  \centering
  \begin{minipage}[b]{.48\textwidth}
    \centering
    \subfloat[\label{fig:l_g_simple}{}]
    {\includegraphics[width=.65\textwidth]{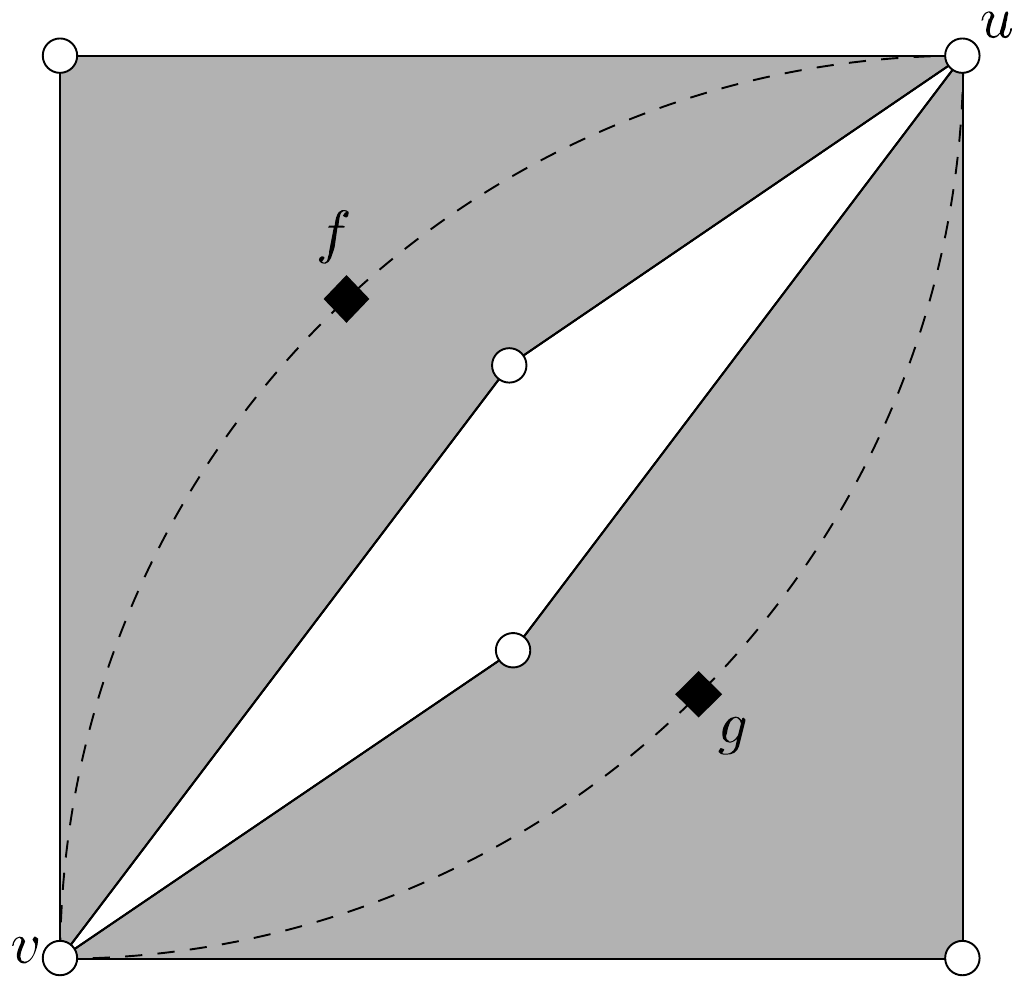}}
  \end{minipage}
  \begin{minipage}[b]{.48\textwidth}
    \centering
    \subfloat[\label{fig:l_g_non_simple}{}]
    {\includegraphics[width=\textwidth]{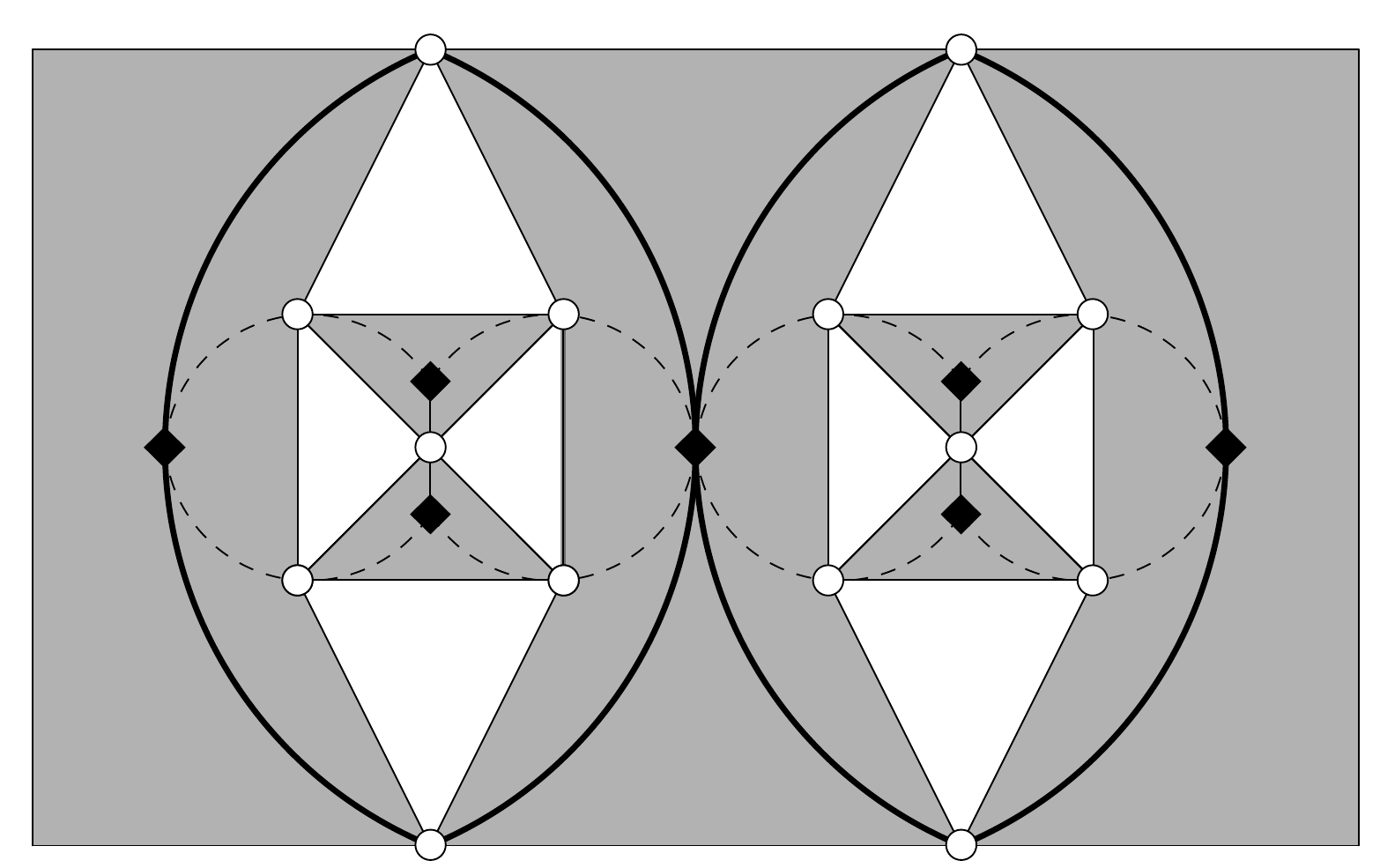}}
  \end{minipage}
  \caption{(a)~Vertices $u$, $f$, $v$ and $g$ define a separating simple closed curve.
  (b)~If $G$ is biconnected, then $IL(G)$ is not necessarily simple (refer to bold edges).}
  \label{fig:l_g}
\end{figure}

We are now ready to state the main theorem of this section.

\begin{theorem}
Every 3-connected 4-regular planar graph admits a realization as a
system of touching circles. \label{thm:3conmain}
\end{theorem}
\begin{proof}
By Lemma~\ref{lem:lg_simple}, $IL(G)$ is simple. So, we can apply
Theorem~\ref{thm:circle-packing} on it. This leads to a drawing in
which each gray-colored face of $G$ corresponds to a circle and two
circles meet at a point if and only if the corresponding faces are
vertex-adjacent (see Fig.\ref{fig:non3connew}). By construction of
$IL(G)$ we have that the vertices of $G$ are the points where
circles touch. Also, every circle contains as arcs all the edges of
the gray face it corresponds to. Since the gray-colored faces
contain all edges of the graph, it follows that the constructed
representation is indeed a system of touching circles for $G$.
\end{proof}

\section{The case of connected 4-regular planar graphs}
\label{sec:connected}

In this section, we will demonstrate an infinite class of connected
4-regular planar graphs that do not admit a realization as a system
of circles. The base of our constructive proof is the octahedron
graph $G_{oct}$ (see Fig.\ref{fig:fig_octahedron0}), which is a
3-connected 4-regular planar graph and hence, by
Theorem~\ref{thm:3conmain}, it admits a realization as a system of
touching circles. Note that in the case where $G$ is not
3-connected, $IL(G)$ is not necessarily simple (for an example refer
to Fig.\ref{fig:l_g_non_simple}). Hence,
Theorem~\ref{thm:circle-packing} cannot be applied directly.

Assume now that a graph $G$ admits a realization, say $R$, as a
system of circles. In general, $R$ is not uniquely defined, as one
can construct infinitely many realizations of $G$ as a system of
circles based on $R$, e.g., by scaling $R$ or by translating $R$
over the plane. With the same spirit, if we slightly change the
radii of the circles or even the centers of the circles of the
realization of the octahedron graph depicted in
Fig.\ref{fig:octahedron_3}, then we can obtain a new realization of
the octahedron graph (again as a system of three mutually crossing
circles), which will be more or less ``\emph{equivalent}'' to the
one depicted in Fig.\ref{fig:octahedron_3}. The same actually holds,
if we simply change the triple of the vertices delimiting its
outerface (as the octahedron graph is symmetric). Intuitively, two
realizations $R$ and $R'$ of $G$ are equivalent if there is a
bijective function from the faces of $R$ to the faces of $R'$, which
maps each face of $R$ to a face of $R'$ of the same \emph{shape},
where the shape of a face is determined by the convexity of the arcs
it consists of, i.e., towards to or away from the interior of the
face.

More formally, given a realization $R$ of $G$ as a system of
circles, first we smooth out $G$ by eliminating vertices of degree
two, and then construct the dual, say $G^*_R$, and orient its edges
as follows: For an edge $e \in E[G]$ incident to two faces $f_e$ and
$f'_e$ of $R$, edge $(f_e,f'_e) \in E[G^*_R]$ is oriented from $f_e$
to $f'_e$ if and only if every straight-line segment with endpoints
on the arc of $R$ corresponding to $e$ does not lie entirely in
$f'_e$. Otherwise, $(f_e,f'_e)$ is directed from $f'_e$ towards
$f_e$. Note that for any circle of $R$, say $c$, a face of $G$ lies
either in the interior of $c$ or in its exterior. Then, orienting
edge $(f_e,f'_e) \in E[G^*_R]$ from $f_e$ to $f'_e$ implies that
$f_e$ lies in the interior of $c$ and $f'_e$ in its exterior, where
$c$ is a circle of $R$ and $e$ is an arc-segment of $c$. We say that
two realizations $R$ and $R'$ of $G$ are \emph{equivalent} if and
only if $G^*_R$ and $G^*_{R'}$ are isomorphic, that is, there is a
bijective function $g:V[G^*_R]\rightarrow V[G^*_{R'}]$ such that
$(\overrightarrow{v, v'}) \in E[G^*_R]$ if and only if
$(\overrightarrow{g(v),g(v')}) \in E[G^*_{R'}]$. Observe that in the
aforementioned definition degree two vertices affect neither the set
of circles that participate in the realization nor their relative
positions. This is the reason why they are omitted. The following
lemma describes all non-equivalent realizations of $G_{oct}$ as a
system of circles.

\begin{lemma}
The octahedron graph has exactly three non-equivalent realizations
as a system of circles, which are shown in
Figures~\ref{fig:octahedron_1}-\ref{fig:octahedron_3}.
\label{lem:octahedron}
\end{lemma}
\begin{proof}
Lemma~\ref{lem:bounds} implies that any realization of the
octahedron graph $G_{oct}$ as a system of circles consists of either
three or four circles. Easy considerations show that in the former
case the three circles are mutually crossing, while in the latter
one the four circles are mutually touching. Consider first the case
of three mutually crossing circles (see Fig.\ref{fig:octahedron_3}),
in which there are six vertices of degree four and every face is a
triangle. Since $G_{oct}$ is the only fully-triangulated 4-regular
planar graph on six vertices, this is indeed a realization of
$G_{oct}$ as a system of circles. There is exactly one face of
$G_{oct}$ that belongs to the interior of all three circles. It
follows that in $G^*_{oct}$ the corresponding vertex has out-degree
three (refer to the innermost vertex of
Fig.\ref{fig:dual_octahedron_3}). Also, for every pair of circles,
there is exactly one face that belongs to the interior of both
circles and to the exterior of the third circle. The corresponding
vertices of $G^*_{oct}$ have out-degree two and in-degree one (refer
to vertices at distance one from the innermost vertex of
Fig.\ref{fig:dual_octahedron_3}). For every circle, there is exactly
one face that belongs to its interior and to the exterior of the
other two circles. The corresponding vertices of $G^*_{oct}$ have
out-degree one and in-degree two (refer to vertices at distance two
from the innermost vertex of Fig.\ref{fig:dual_octahedron_3}).
Finally, the vertex corresponding to the outerface has in-degree
three. From the above, it follows that the oriented dual $G^*_{oct}$
that corresponds to a realization of $G_{oct}$ as a system of three
mutually crossing circles is isomorphic to the one given in
Fig.\ref{fig:dual_octahedron_3}, where the vertex corresponding to
the outerface is omitted.

\begin{figure}[t]
  \centering
  \begin{minipage}[b]{.32\textwidth}
    \centering
    \subfloat[\label{fig:dual_octahedron_3}{}]
    {\includegraphics[width=\textwidth]{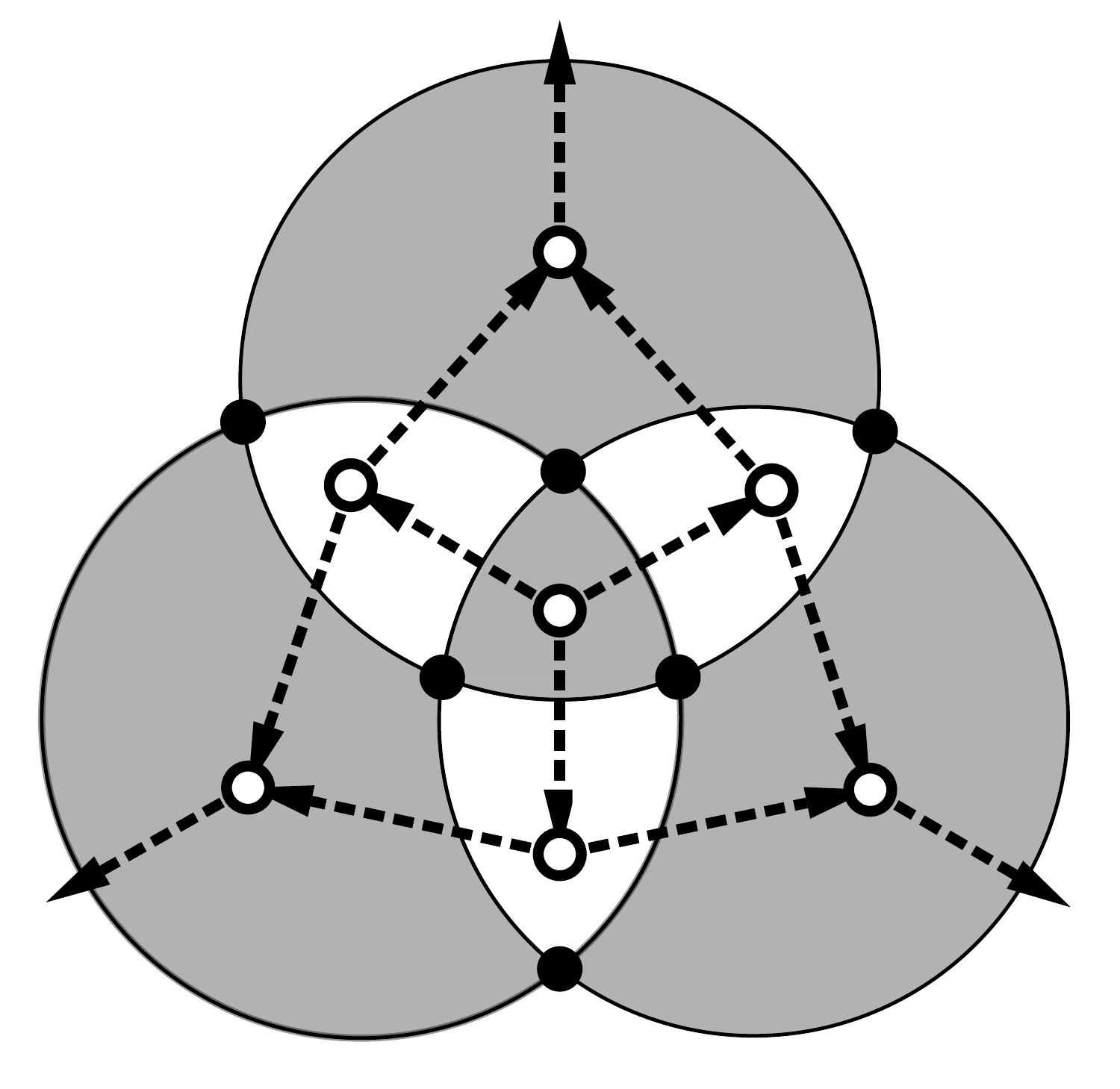}}
  \end{minipage}
  \begin{minipage}[b]{.32\textwidth}
    \centering
    \subfloat[\label{fig:dual_octahedron_2}{}]
    {\includegraphics[width=\textwidth]{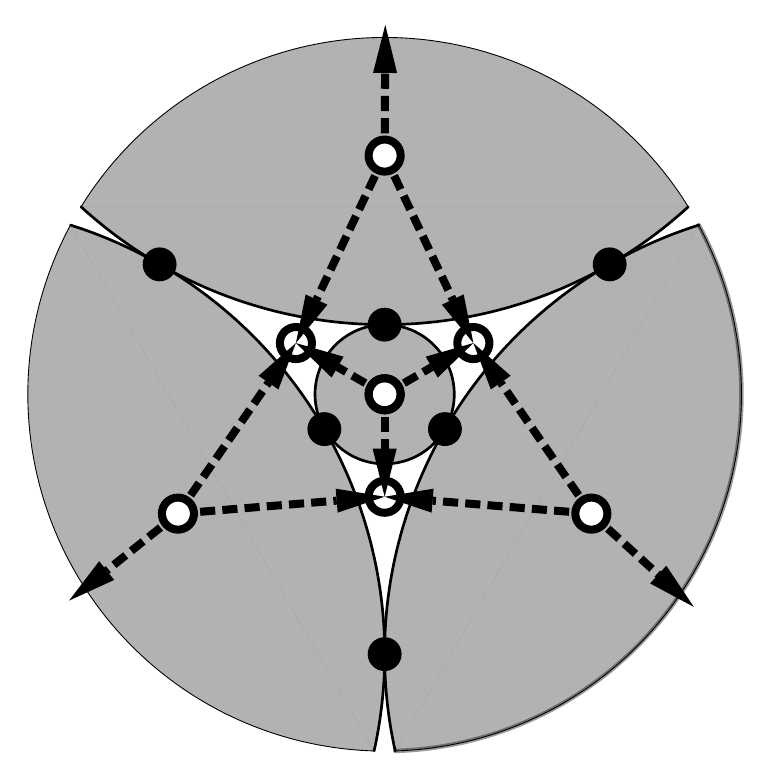}}
  \end{minipage}
    \begin{minipage}[b]{.32\textwidth}
    \centering
    \subfloat[\label{fig:dual_octahedron_1}{}]
    {\includegraphics[width=\textwidth]{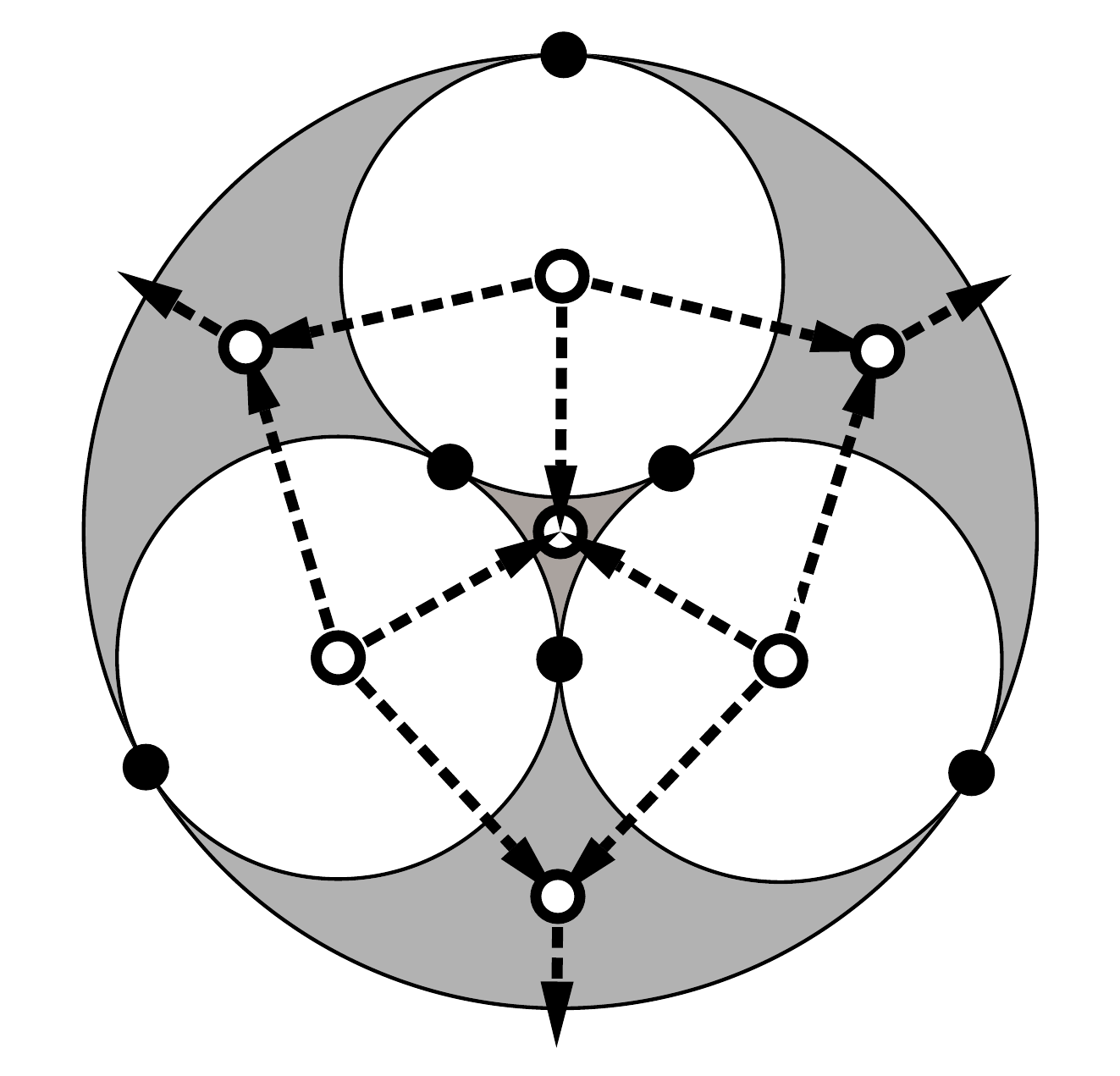}}
  \end{minipage}
  \caption{Illustration of the oriented duals (drawn in dotted) of $G_{oct}$ that correspond to the three non-equivalent realizations of the octahedron graph $G_{oct}$ as system of circles.}
  \label{fig:dual_octahedron}
\end{figure}

In the case of four touching circles, by Lemma~\ref{lem:bounds} it
follows that every circle has exactly three vertices. Since vertices
are defined by touching points and two circles can have at most one
common point, it follows that a circle touches all three other
circles of the representation. Let $c_1,\ldots,c_4$  be the four
circles. If $c_2$ lies in the interior of $c_1$, then $c_3$ and
$c_4$ are also in the interior of $c_1$ (since they touch with
$c_2$). This implies that either all circles have empty interiors or
one circle contains all three other circles in its interior (see
Figures~\ref{fig:octahedron_2} and~\ref{fig:octahedron_1},
respectively), in which there are six vertices of degree four
defined, and every face is a triangle. Since $G_{oct}$ is the only
fully-triangulated 4-regular planar graph on six vertices, these are
indeed realizations of $G_{oct}$ as system of circles. For the case
where all circles have empty interior, circles $c_1$ to $c_4$ are
faces of $G_{oct}$ and the corresponding vertices of $G^*_{oct}$
have out-degree three (refer to the innermost vertex of
Fig.\ref{fig:dual_octahedron_2} and to vertices at distance two from
it). All other faces of $G_{oct}$ lie in the exterior of all circles
and have therefore in-degree three and out-degree zero (refer to
vertices at distance one from the innermost vertex of
Fig.\ref{fig:dual_octahedron_2}; the outerface vertex also fits to
this case). Therefore, the oriented dual $G^*_{oct}$ that
corresponds to a realization of $G_{oct}$ as four touching circles
with disjoint interiors is isomorphic to the one given in
Fig.\ref{fig:dual_octahedron_2}, where the vertex corresponding to
the outerface is omitted.

Suppose now that one circle, say circle $c_1$, contains in its
interior all three other circles. Then, circles $c_2$, $c_3$ and
$c_4$ are faces of $G_{oct}$ and the corresponding vertices of
$G^*_{oct}$ have out-degree three (refer to vertices at distance one
from the innermost vertex of Fig.\ref{fig:dual_octahedron_1}). There
is exactly one face delimited by $c_2$, $c_3$ and $c_4$ (in the
interior of $c_1$) that corresponds to a vertex of $G^*_{oct}$ with
in-degree three (refer to the innermost vertex of
Fig.\ref{fig:dual_octahedron_1}). Three distinct faces share an edge
with $c_1$ and lie in its interior corresponding to three vertices
of $G^*_{oct}$ with in-degree two and out-degree one (refer to
vertices at distance two from the innermost vertex of
Fig.\ref{fig:dual_octahedron_1}). The outerface corresponds to a
vertex of in-degree three. Similarly to the previous cases, the
oriented dual $G^*_{oct}$ that corresponds to a realization of
$G_{oct}$ as four touching circles where one circle contains in its
interior all three other circles, is isomorphic to the one given in
Fig.\ref{fig:dual_octahedron_1}, where the vertex corresponding to
the outerface is omitted.

It is not hard to see that the three realizations presented are not
equivalent: In the realization of Fig.\ref{fig:dual_octahedron_3}
there is only one vertex with out-degree three, while in the
realization of Fig.\ref{fig:dual_octahedron_2} there are four
vertices with out-degree three, and in the realization of
Fig.\ref{fig:dual_octahedron_1} there are three. Since the degree
sequences of the three dual graphs are different they can't be
isomorphic.

Any realization of the octahedron graph as a system of circles will
use either three mutually crossing circles or four touching circles.
Since we proved that the oriented dual in the former case is
isomorphic to the one of Fig.\ref{fig:dual_octahedron_3}, while in
the latter case isomorphic either to the one of
Fig.\ref{fig:dual_octahedron_2} or to the one of
Fig.\ref{fig:dual_octahedron_1}, it follows that it is always
isomorphic to one of the digraphs of Fig.\ref{fig:dual_octahedron},
giving a total of three non-equivalent realizations of $G_{oct}$ as
a system of circles.
\end{proof}

\begin{figure}[t]
  \centering
  \begin{minipage}[b]{.48\textwidth}
    \centering
    \subfloat[\label{fig:gadget-special}{The gadget-subgraph.}]
    {\includegraphics[width=.86\textwidth]{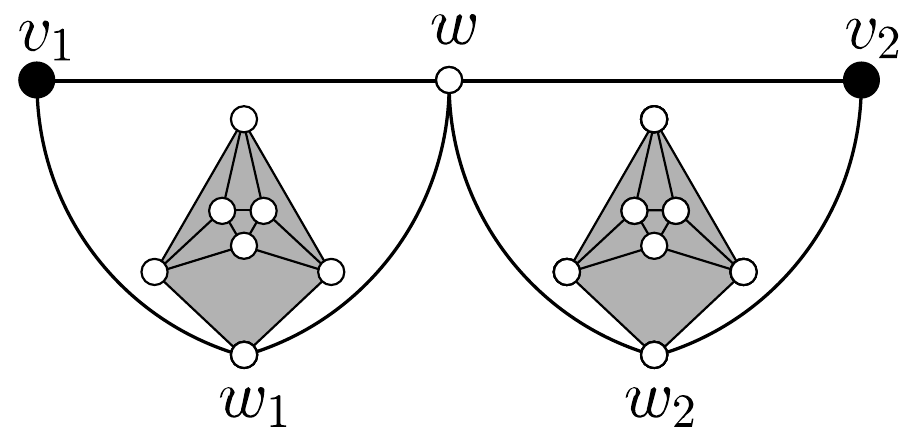}}
  \end{minipage}
  \begin{minipage}[b]{.48\textwidth}
    \centering
    \subfloat[\label{fig:gadget-abstract}{Abstraction of the gadget-subgraph.}]
    {\includegraphics[width=.86\textwidth]{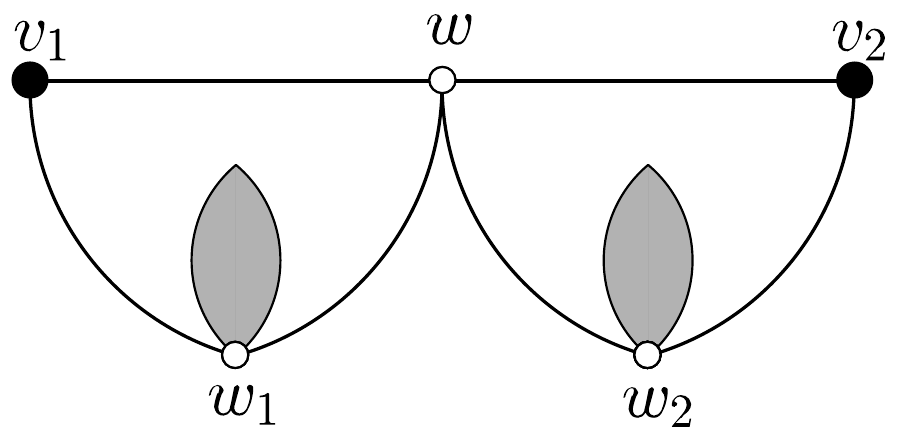}}
  \end{minipage}
  \caption{Illustrations of the gadget-subgraph.}
  \label{fig:gadget}
\end{figure}

Initially, we will exhibit a specific connected 4-regular planar
graph that does not admit a realization as a system of circles. This
graph will be constructed based on $G_{oct}$, augmented by
appropriately ``attaching'' a specific \emph{gadget-subgraph} to its
edges, leading thus to a graph, say $G_{oct}^{aug}$, that contains
cut-vertices and separation pairs (note that any connected 4-regular
planar graph is bridgeless \cite[pp.34]{WE00}). The gadget-subgraph
is illustrated in Fig.\ref{fig:gadget-special}. Observe that it
contains exactly two vertices of degree two, namely $v_1$ and $v_2$,
which are its \emph{endpoints}. Now we replace every edge $e=(u,v)$
of $G_{oct}$ by a path consisting of 8 internal vertices. Clearly,
the graph, say $G^{sub}_{oct}$, that is obtained in this manner is a
subdivision of $G_{oct}$. Let $u \rightarrow z_1 \rightarrow z_2
\rightarrow \dots \rightarrow z_8 \rightarrow v$ be the path
replacing edge $(u,v)$. We associate four copies of the
gadget-subgraph having vertices $z_1,\dots,z_8$ as their endpoints:
the first gadget-subgraph connects vertices $z_1$ and $z_6$, the
second connects $z_2$ and $z_5$, the third connects $z_3$ and $z_8$
and the last connects $z_4$ and $z_7$ (see Fig.\ref{fig:edge_split},
in which the gadget-subgraphs are drawn with dashed curves joining
the end-vertices; Fig.\ref{fig:counter} depicts the resulting graph
$G_{oct}^{aug}$).

\begin{figure}[b]
  \centering
  \begin{minipage}[b]{.54\textwidth}
    \centering
    \subfloat[\label{fig:edge_split}{Attaching gadget-subgraphs on $e=(u,v)$.}]
    {\includegraphics[width=\textwidth]{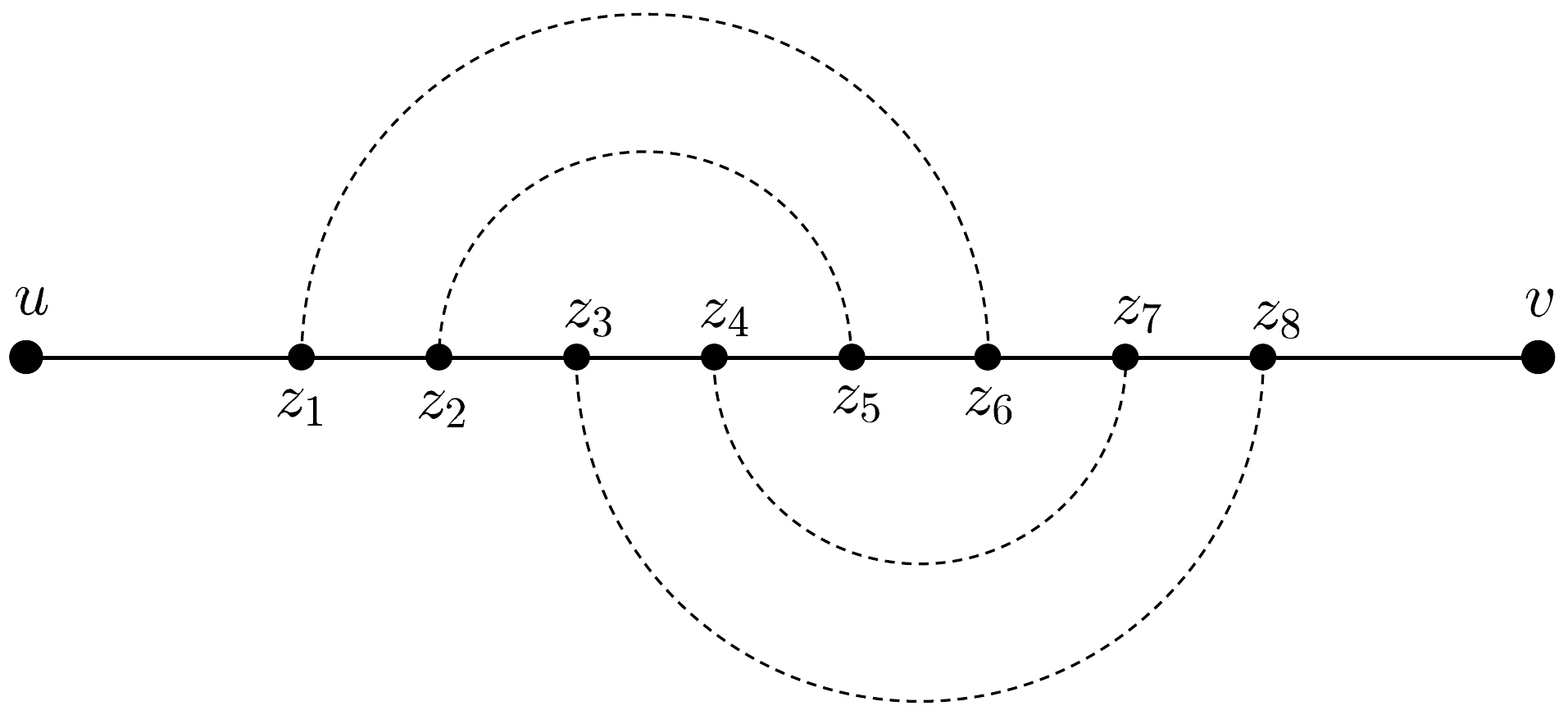}}
  \end{minipage}
  \begin{minipage}[b]{.44\textwidth}
    \centering
    \subfloat[\label{fig:counter}{The resulting graph $G_{oct}^{aug}$.}]
    {\includegraphics[width=.8\textwidth]{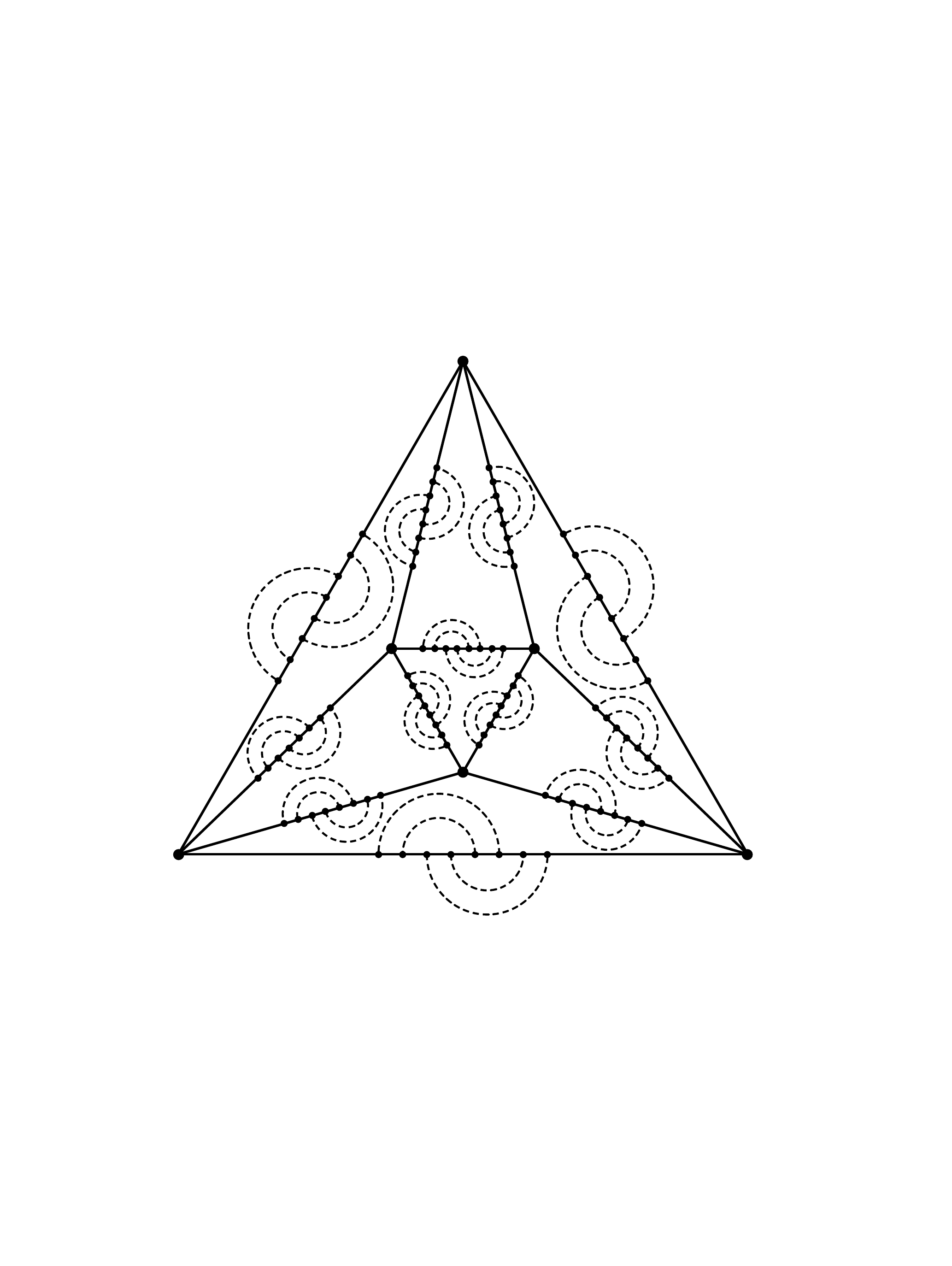}}
  \end{minipage}
  \caption{Each dashed edge corresponds to the gadget-subgraph of Fig.\ref{fig:gadget}.}
  \label{fig:edge_split_counter}
\end{figure}

The \emph{skeleton} of the gadget-subgraph consists of vertices
$v_1$, $v_2$, $w_1$, $w_2$ and $w$ (see Fig.\ref{fig:gadget}) and
edges $(v_i,w)$, $(w_i,w)$ and $(v_i,w_i)$, $i=1,2$. If we remove
the edges of the skeleton, the remaining graph consists of three
isolated vertices (namely $v_1$, $v_2$ and $w$) and two disjoint
graphs that are subdivisions of the octahedron graph (refer to the
gray-shaded graphs of Fig.\ref{fig:gadget-special}). In this section
we will exhibit some properties of the gadget-subgraphs. These
properties are not actually due to the structure of the
gadget-subgraphs. In fact, any graph in which every vertex has
degree four except for exactly one degree-2 vertex on the outerface
can be used instead of the gadget-subgraphs still guaranteeing the
same properties. The general situation is shown in
Fig.\ref{fig:gadget-abstract}, where the subgraphs are drawn as
self-loops at vertices $w_1$ and $w_2$.  For convenience, we will
refer to these subgraphs as \emph{loop-subgraphs}.

\begin{lemma}
Let $G$ be a 4-regular planar graph that contains at least one copy
of the gadget-subgraph. Suppose that there is a realization of $G$
as a system of circles. Then,  the skeleton of each gadget-subgraph
in this realization consists of two circles $C_1$ and $C_2$ tangent
at a point $w$, where circle $C_i$ contains vertices $\{v_i,w_i,w\}$
and the arc-segments realizing edges $(v_i, w_i)$, $(w_i, w)$,
and $(v_i, w)$, for $i=1,2$. \label{lem:gadget-realization}
\end{lemma}
\begin{proof}
Suppose that there is a realization of $G$ as a system of circles
and consider a copy of the gadget-subgraph in this realization.
Since every vertex is defined by exactly two circles and $w_i$ is a
cut-vertex, it follows that one of the two circles defining $w_i$
contains vertices that belong only to its loop-subgraph, $i=1,2$.
Hence, the edges $(v_i,w_i)$ and $(w_i,w)$ belong to the same
circle, $i=1,2$. Let $C_i$ be the circle that contains $(v_i,w_i)$
and $(w_i,w)$ and $C'_i$ the circle that contains $(v_i,w)$,
$i=1,2$. We claim that $C_i = C'_i$, $i=1,2$. Observe that this
implies the lemma. For the sake of contradiction, assume that $C_1
\neq C'_1$. Since vertex $w$ is defined by exactly two circles, we
have that $\left\{C_1,C'_1\right\}=\left\{C_2,C'_2\right\}$, which
also implies that $C_2\neq C'_2$. Then, $C_1$ and $C'_1$ have at
least three points in common, namely vertices $v_1$, $v_2$ and $w$,
from which we obtain $C_1 = C'_1$; a contradiction.
\end{proof}

\begin{figure}[t]
  \centering
  \begin{minipage}[b]{.43\textwidth}
    \flushleft
    {\includegraphics[width=.9\textwidth]{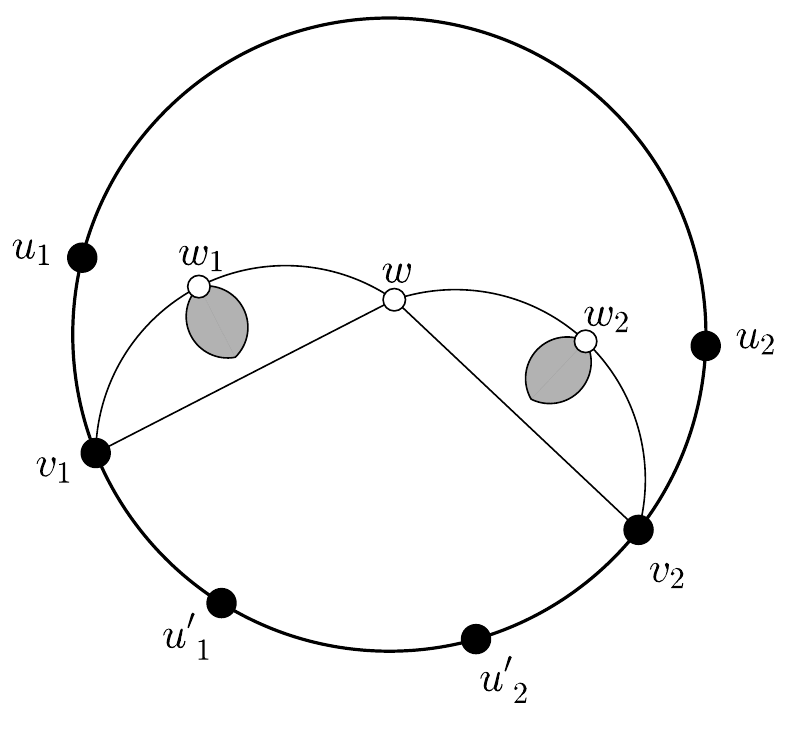}}
  \end{minipage}
    \caption{Configuration considered in proof of Lemma~\ref{lem:independent}.}
    \label{fig:gadget_plus}
\end{figure}

\begin{lemma}
Let $G$ be a 4-regular planar graph and $G^{sub}$ a subdivision of
$G$. Let $ v_1$ and  $v_2$ be two subdivision vertices of $G^{sub}$,
i.e. $v_1$ and $v_2$ are degree-2 vertices. Attach a
gadget-subgraph, so that $v_1$ and $v_2$ are its endpoints, and such
that the resulting graph is planar. Then, in any realization of the
resulting graph as a system of circles, the realization of the
gadget-subgraph and the realization of $G^{sub}$ are
\emph{independent}, i.e., any circle contains edges that belong
exclusively either to the gadget-subgraph or to $G^{sub}$.
\label{lem:independent}
\end{lemma}
\begin{proof}
Refer to Fig.\ref{fig:gadget_plus}. Since the resulting graph is
planar, $v_1$ and $v_2$ lie on the boundary of a face of $G^{sub}$.
By Lemma \ref{lem:gadget-realization}, it follows that edges
$(v_i,w_i)$, $(v_i,w)$ and $(w_i,w)$ belong to the same circle, say
$C_i$, $i=1,2$. Let $u_j$, $u'_j$ be the two neighbors of vertex
$v_j$ in $G^{sub}$, $j=1,2$ (Note that $u'_1=v_2$ and $u'_2=v_1$ are
possible). Since every vertex belongs to exactly two circles, it
follows that edges $(u_j,v_j)$ and $(v_j,u'_j)$ belong to a circle
different from $C_1$ and $C_2$. Therefore, if we remove $C_1$ and
$C_2$ and the circles representing the loop-subgraphs of the
gadget-subgraph, we have a representation of the remaining graph
(namely of graph $G^{sub}$), as a system of circles.
\end{proof}

From the above, it follows that in any realization of
$G_{oct}^{aug}$ as a system of circles the realization of each
gadget-subgraph and the realization of $G^{sub}_{oct}$ are
independent. This is summarized in the following corollary.

\begin{corollary}
In any realization of $G_{oct}^{aug}$ as a system of circles, the
realization of each gadget-subgraph and the realization of
$G^{sub}_{oct}$ are independent. \label{cor:independent-goct}
\end{corollary}

Corollary~\ref{cor:independent-goct} is the key element of our
proof. It implies that the realization of $G^{sub}_{oct}$ obtained
from a realization of $G^{aug}_{oct}$ by removing all vertices of
the gadget-subgraphs except for their endpoints will be equivalent
to one of the realizations of $G_{oct}$ depicted in
Figures~\ref{fig:octahedron_1}-\ref{fig:octahedron_3} (recall that
the definition of two equivalent realizations ignores vertices of
degree-2).

\begin{figure}[t]
  \centering
  \begin{minipage}[b]{.53\textwidth}
    \centering
    \subfloat[\label{fig:gadget-circles_a}{}]
    {\includegraphics[width=\textwidth]{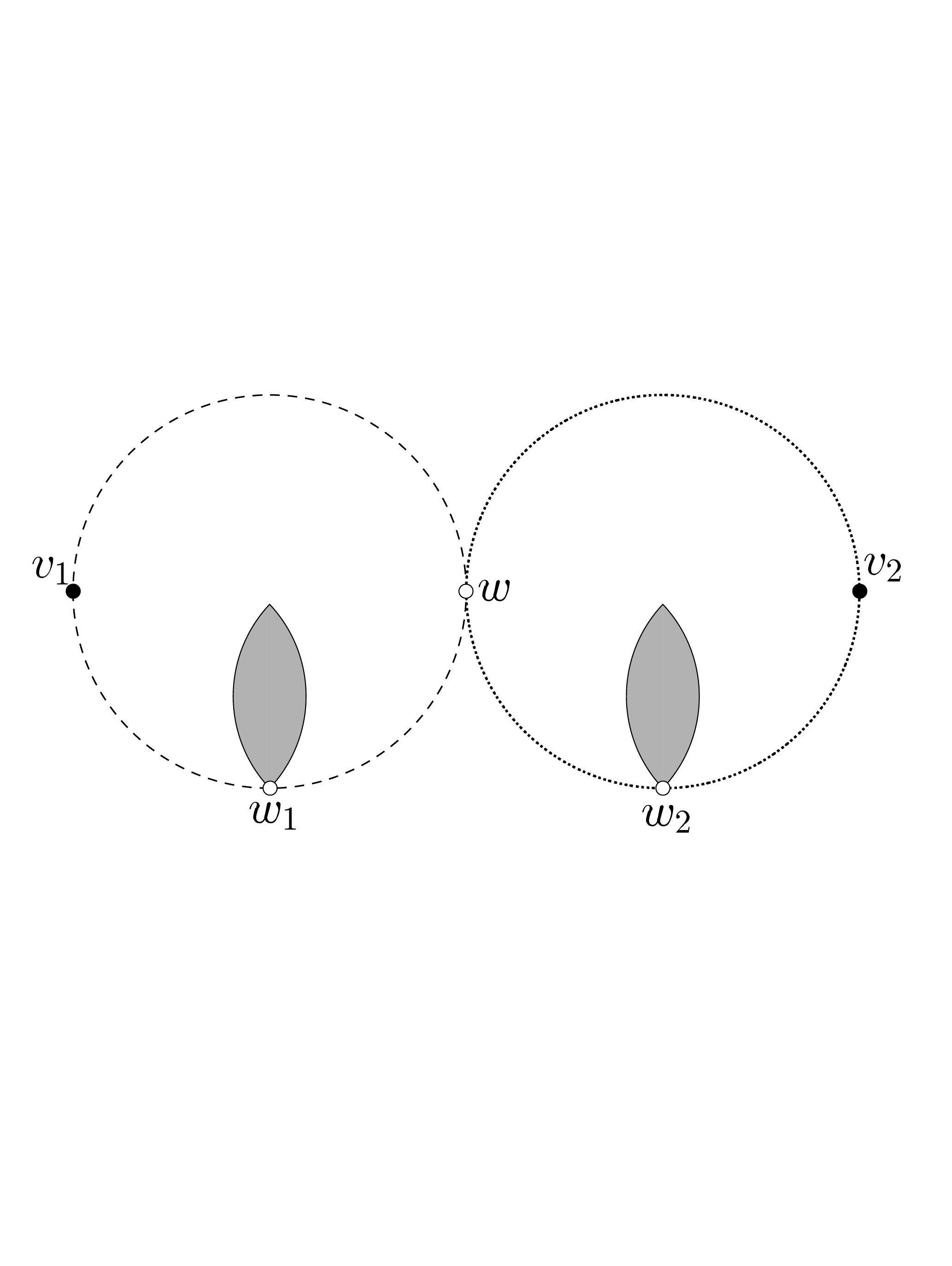}}
  \end{minipage}
  \begin{minipage}[b]{.33\textwidth}
    \centering
    \subfloat[\label{fig:gadget-circles_b}{}]
    {\includegraphics[width=\textwidth]{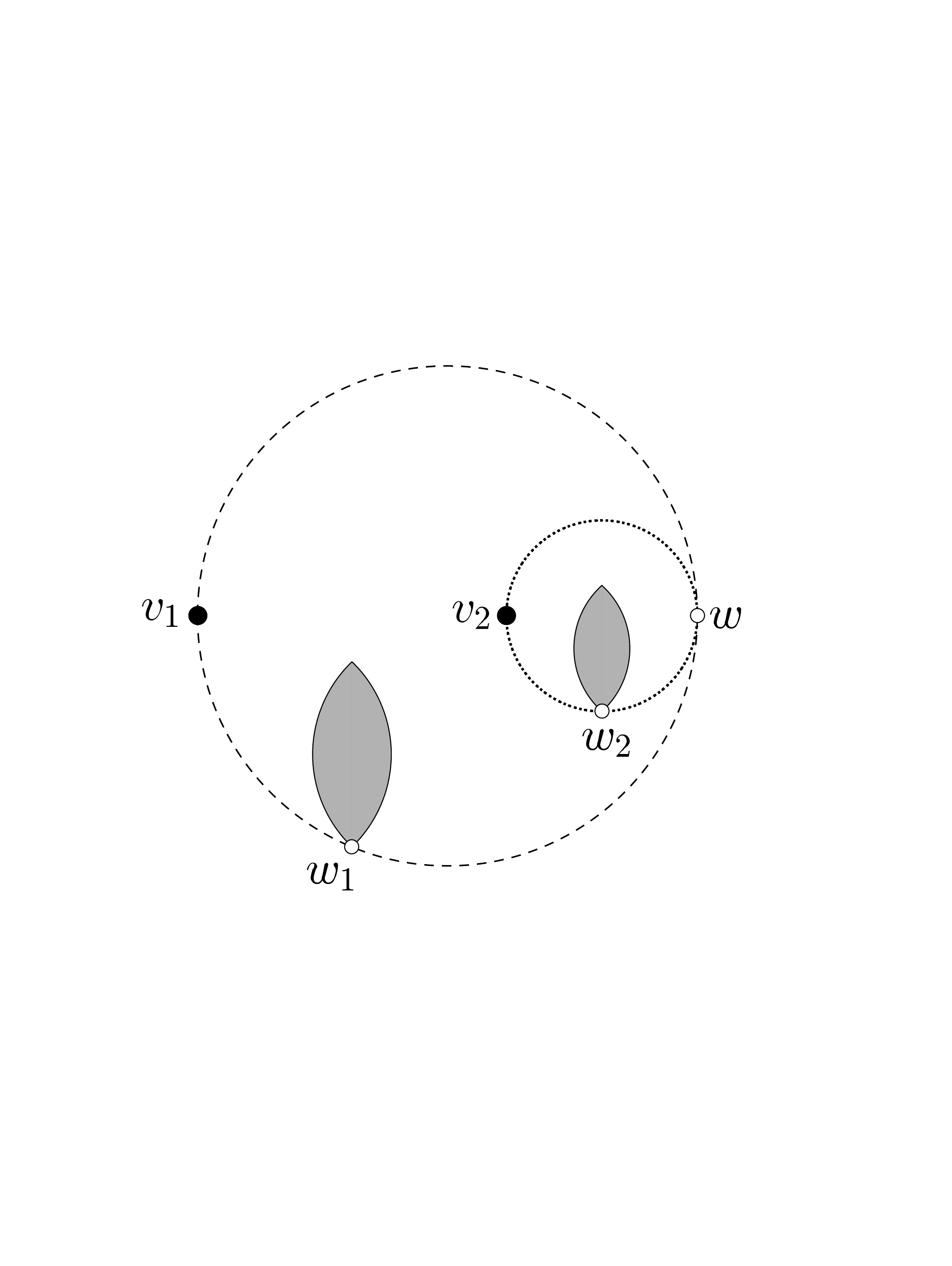}}
  \end{minipage}
  \caption{Non-equivalent realizations of the gadget-subgraph as system of circles: (a)~the two circles representing the skeleton have disjoint interior, and (b)~one circle lies in the interior of the other.}
  \label{fig:gadget-circles}
\end{figure}

In any planar embedding of $G_{oct}$,  there is always a triangular
face that shares no vertex and no edge with the outerface. Hence,
the gadget-subgraphs attached to the edges of this triangular face
have to be realized as in Fig.\ref{fig:gadget-circles_a} (in fact a
realization as in Fig.\ref{fig:gadget-circles_b} is only possible if
the gadget-subgraph is incident to the outerface). So, there is a
total of six gadget-subgraphs attached along the edges of this
triangular face, each with a realization as in
Fig.\ref{fig:gadget-circles_a}.

In the following, we state two useful geometric results regarding
tangent circles. We denote by $C(O,r)$ a circle with center $O$ and
radius $r$.

\begin{lemma}
Let $C_1(O_1,r_1)$ and $C_2(O_2,r_2)$ be two circles, so that $C_1$
is tangent to $C_2$ at point $p_1$ and $C_2$ lies entirely in the
interior of $C_1$. Let $C(O,r)$ be another circle that is tangent to
$C_1$ at point $p_2$ ($p_2\neq p_1$), tangent to $C_2$ and lies in
the interior of $C_1$ (see Fig.\ref{fig:circle-lemma1}). If $\phi$
is the angle $\widehat{p_1O_1p_2}$, then the radius $r$ of $C$ is an
increasing function of $\phi$, $\phi\in(0,\pi]$.
\label{lem:circle-lemma1}
\end{lemma}
\begin{proof}
W.l.o.g., we assume that $O_1$ coincides with the origin of the
Cartesian coordinate system and point $p_1$ lies on the x-axis,
i.e., at point $(r_1,0)$. Then, the center of circle $C_2$ is at
point $(r_1-r_2, 0)$, while the center of circle $C$ is at point
$((r_1-r)\cos\phi,(r_1-r)\sin\phi)$, as shown in
Fig.\ref{fig:circle-lemma1}. Since $C_2$ and $C$ are tangent the
distance between their centers equals to the sum of their radii,
i.e.:

\begin{tabular}{ll}
  ~ & $[(r_1-r)\cos\phi-(r_1-r_2)]^2+[(r_1-r)\sin\phi]^2=(r_2+r)^2$ \\
  $\Rightarrow$ & $(r_1-r_2)^2+(r_1-r)^2-2(r_1-r_2)(r_1-r)\cos\phi=(r_2+r)^2$ \\
  $\Rightarrow$ & $(r_1+r_2)(r_1-r)-2r_1r_2-(r_1-r_2)(r_1-r)\cos\phi=0$ \\
  $\Rightarrow$ & $r=r_1-\frac{2r_1r_2}{r_1+r_2-(r_1-r_2)\cos\phi}$ \\
\end{tabular}

By the above equation, when $\phi$ is increasing in the interval
$(0,\pi]$, $\cos\phi$ is decreasing and $r$ is increasing. Hence,
circle $C$ has maximum radius for angle $\phi=\pi$.
\end{proof}

\begin{figure}[t]
  \centering
  \begin{minipage}[b]{.48\textwidth}
    \centering
    \subfloat[\label{fig:circle-lemma1}{}]
    {\includegraphics[width=.8\textwidth]{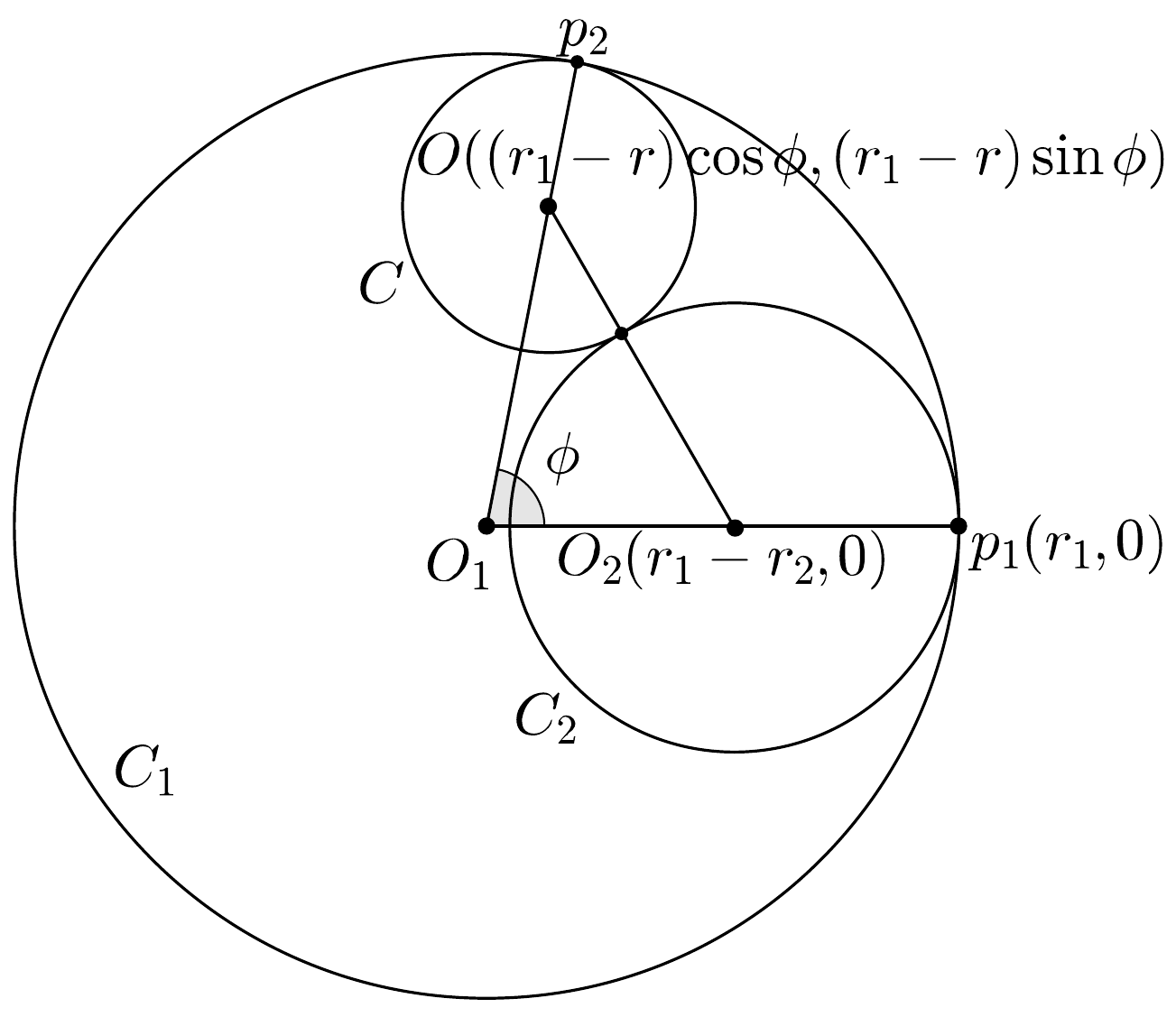}}
  \end{minipage}
  \begin{minipage}[b]{.48\textwidth}
    \centering
    \subfloat[\label{fig:circle-lemma2}{}]
    {\includegraphics[width=.8\textwidth]{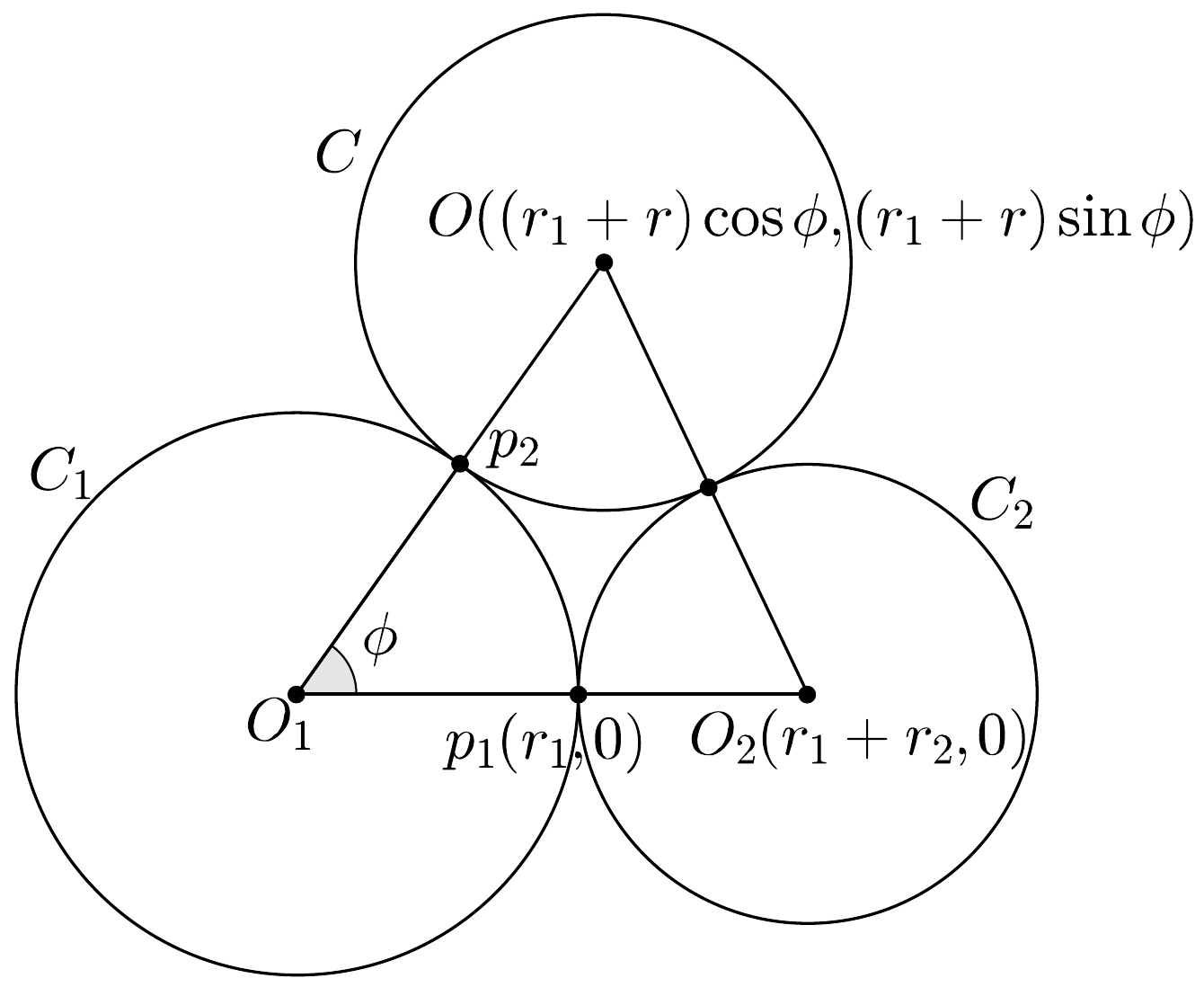}}
  \end{minipage}
  \caption{(a)~Configuration considered in Lemma~\ref{lem:circle-lemma1}. (b)~Configuration considered in Lemma~\ref{lem:circle-lemma2}.}
  \label{fig:circle-lemma12}
\end{figure}

A similar result holds if circles $C$ and $C_2$ lie outside circle
$C_1$.

\begin{lemma}
Let $C_1(O_1,r_1)$ and $C_2(O_2,r_2)$ be two circles, so that $C_1$
is tangent to $C_2$ at point $p_1$ and $C_1$ lies entirely in the
exterior of $C_2$. Let $C(O,r)$ be another circle that is tangent to
$C_1$ at point $p_2$ ($p_2\neq p_1$), tangent to $C_2$ and lies in
the exterior of $C_1$ (see Fig.\ref{fig:circle-lemma2}). If $\phi$
is the angle $\widehat{p_1O_1p_2}$, then the radius $r$ of $C$ is an
increasing function of $\phi$,
$\phi\in(0,\arccos(\frac{r_1-r_2}{r_1+r_2})]$.
\label{lem:circle-lemma2}
\end{lemma}
\begin{proof}
The proof of Lemma \ref{lem:circle-lemma2} is similar to the one of
Lemma \ref{lem:circle-lemma1}. So, we omit the details. We simply
mention the corresponding equation for $r$, which is the following:
$$r=\frac{2r_1r_2}{r_2-r_1+(r_2+r_1)\cos\phi}-r_1$$
Note that circle $C$ does not always exist. For an example, refer to
Fig.\ref{fig:circle-lemma2}, when $\phi=\pi/2$ and $r_1 > r_2$. In
particular, for given radii $r_1$ and $r_2$, angle $\phi$ is bounded
from above by value $\arccos(\frac{r_1-r_2}{r_1+r_2})$, which
corresponds to the angle in the extreme case where circle $C$ is of
infinite radius and is therefore reduced to the common tangent of
circles $C_1$ and $C_2$.
\end{proof}

\begin{lemma}
Consider a circle $C(O,r)$ and an arc $\arc{AB}$ of $C$ with
$\widehat{AOB}=\phi<\pi$. Let $C_1(O_1,r_1)$ and $C_2(O_2,r_2)$ be
two tangent circles, that are both tangent to $C$ at points $A$ and
$B$ respectively (see Fig.\ref{fig:circle-lemma}). Let
$C'_1(O'_1,r'_1)$ and $C'_2(O'_2,r'_2)$ be another such pair of
tangent circles that are both tangent to $C$ at points $A'$ and $B'$
respectively (with $A'$ and $B'$ on the arc between $A$ and $B$), so
that the two pairs of circles have no crossing and no touching
points. Assuming that $C_i$ and $C_i'$, $i=1,2$ are all in the
interior of $C$ or in the exterior of $C$ and $A$, $A'$, $B'$, and
$B$ occur in this order on the arc between $A$ and $B$, then:
$$|\arc{A'B'}|<|\arc{AA'}|~and~|\arc{A'B'}|<|\arc{B'B}|$$
\label{lem:circle-lemma}
\end{lemma}
\begin{proof}
Consider the circle $C_1$. There exist two circles $C''_i$ for
$i=1,2$ with radius $r''_i$, that are both tangent to $C_1$ and also
tangent to $C$ at points $A'$ and $B'$ respectively. Note that for
$i=1,2$ circle $C''_i$ contains circle $C'_i$, and so $r''_i\geq
r'_i$. We have that $\widehat{AOA'}<\widehat{AOB'}<\widehat{AOB}$
and therefore, by Lemmas~\ref{lem:circle-lemma1} and
\ref{lem:circle-lemma2}, $r'_i\leq r''_i<r_2$ for $i=1,2$.
Similarly, starting from circle $C_2$ we have that $r'_i\leq
r''_i<r_1$ for $i=1,2$. Now, for circle $C'_1$ we have that
$r'_2<r_1$ and $\widehat{AOA'},\widehat{A'OB'}<\pi$. Then by
Lemmas~\ref{lem:circle-lemma1} and \ref{lem:circle-lemma2} it
follows that $\widehat{A'OB'}<\widehat{AOA'}$. Similarly, we
conclude that $\widehat{A'OB'}<\widehat{B'OB}$. So, we have that:
$$|\arc{A'B'}|<|\arc{AA'}|~and~|\arc{A'B'}|<|\arc{B'B}|.$$
\end{proof}

\begin{figure}[t!]
  \centering
  \begin{minipage}[b]{.48\textwidth}
    \centering
    \subfloat[\label{fig:fig_circle_lemma_a}{}]
    {\includegraphics[width=.8\textwidth]{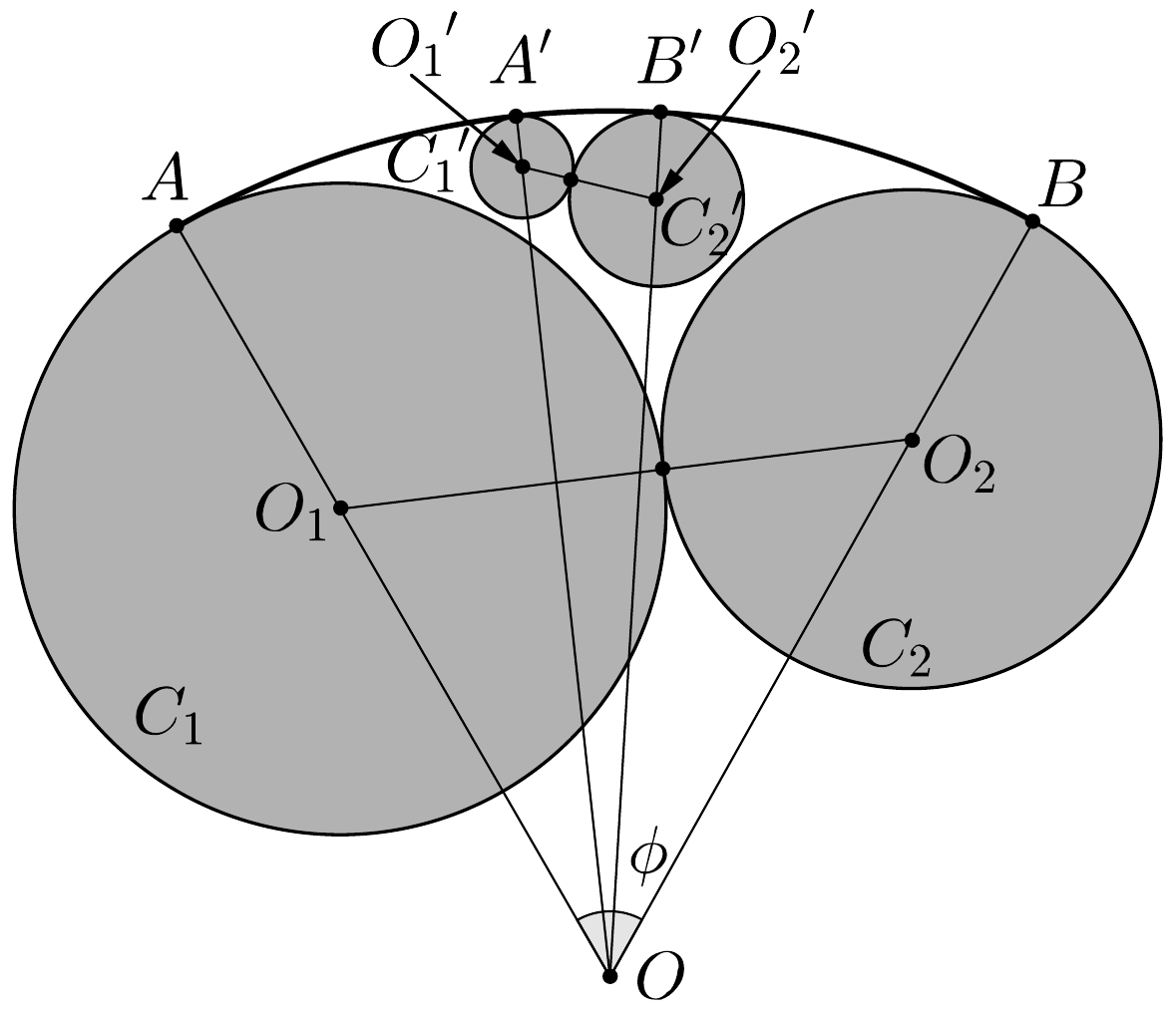}}
  \end{minipage}
  \begin{minipage}[b]{.48\textwidth}
    \centering
    \subfloat[\label{fig:fig_circle_lemma_b}{}]
    {\includegraphics[width=.7\textwidth]{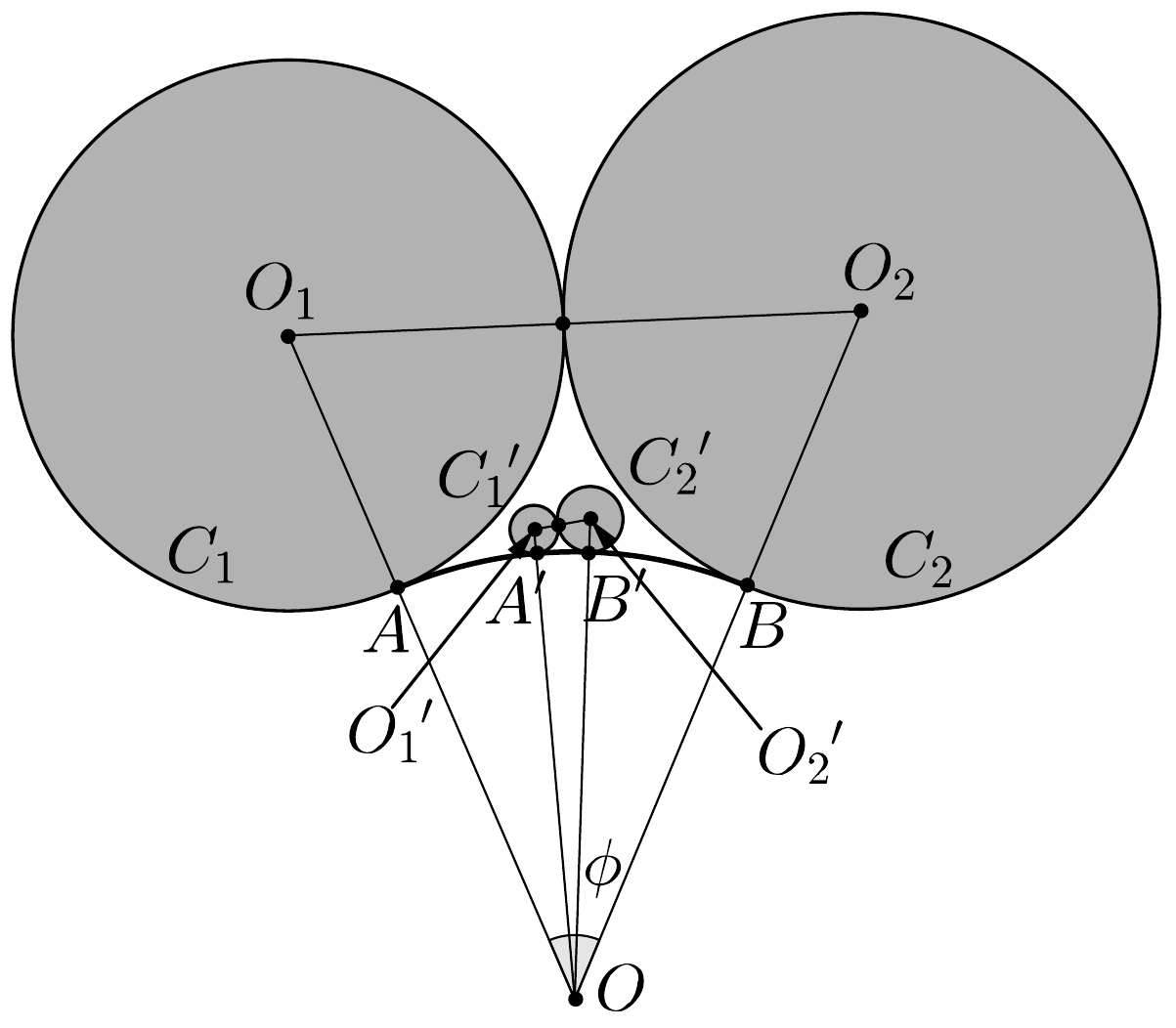}}
  \end{minipage}
  \caption{Configurations used in Lemma~\ref{lem:circle-lemma}.}
  \label{fig:circle-lemma}
\end{figure}

Note that Lemma~\ref{lem:circle-lemma} is still true when the four
circles lie either in the interior or on the exterior of circle $C$.
Let $e=(u,v)$ be an edge of the innermost interior face of the
octahedron graph $G_{oct}$, as shown in Fig.\ref{fig:counter}.
Assume that in a realization of the octahedron as a system of
circles, $e$ is drawn as an arc of a circle $C(O,r)$, with
$\widehat{uOv}=\phi<\pi$. The next lemma proves that this assumption
leads to a contradiction to the existence of a realization of graph
$G_{oct}^{aug}$ as a system of circles.

\begin{lemma}
Consider a circle $C(O,r)$ and assume that edge $e=(u,v)\in
E[G_{oct}]$ is drawn as an arc segment $\arc{uv}$ of $C$ such that
$\widehat{uOv}=\phi<\pi$. If we attach two pairs of gadget-subgraphs
along $e$, as shown in Fig.\ref{fig:edge_split}, then the resulting
subgraph of $G_{oct}^{aug}$ does not admit a realization as a system
of circles.\label{lem:circle-basic}
\end{lemma}

\begin{proof}
By Lemma~\ref{lem:gadget-realization} we have that each
gadget-subgraph is drawn as a pair of tangent circles that are also
tangent to the arc $\arc{uv}$ at points $z_i, z_j$. Furthermore, the
two tangent circles have disjoint interiors since they are not on
the outerface of $G_{oct}^{aug}$, as in
Fig.\ref{fig:gadget-circles_a}. By applying
Lemma~\ref{lem:circle-lemma} to the pair of gadget-subgraphs with
endpoints $z_1$, $z_6$ and $z_2$, $z_5$, we have:
$$\arc{z_2z_5}<\arc{z_1z_2}~and~\arc{z_2z_5}<\arc{z_5z_6}$$
Similarly, for the pair of gadget-subgraphs with endpoints $z_3$,
$z_8$ and $z_4$, $z_7$, we have:
$$\arc{z_4z_7}<\arc{z_3z_4}~and~\arc{z_4z_7}<\arc{z_7z_8}$$
Combining those inequalities and the fact that
$\arc{z_iz_j}\leq\arc{z_{i'}z_{j'}}$ for $i'\leq i\leq j\leq j'$, we
have:
$$\arc{z_4z_7}<\arc{z_3z_4}\leq
\arc{z_2z_5}<\arc{z_5z_6}\leq\arc{z_4z_7}$$
that is $\arc{z_4z_7}<\arc{z_4z_7}$, which is a contradiction.
\end{proof}

In order to complete the proof that graph $G_{oct}^{aug}$ does not
admit a realization as a system of circles, it suffices to show that
in any realization of $G_{oct}$ (and therefore of $G_{oct}^{sub}$),
at least one edge of the innermost interior face meets the
requirements of Lemma~\ref{lem:circle-basic}.

\begin{lemma}
In any realization of the octahedron as a system of circles, at
least one edge, say $e=(u,v)$, of the innermost interior face is
drawn as an arc of a circle $C(O,r)$ so that
$\widehat{uOv}=\phi<\pi$. \label{lem:face}
\end{lemma}
\begin{proof}
By Lemma~\ref{lem:octahedron}, it suffices to show that the lemma
holds for the three non-equivalent representations shown in
Fig.\ref{fig:octahedron}. For the first two representations of
Fig.\ref{fig:octahedron}, the result is almost straightforward. More
precisely, let $C_1(O_1,r_1)$, $C_2(O_2,r_2)$ and $C_3(O_3,r_3)$ be
the circles (white-colored in Fig.\ref{fig:octahedron_1}) that
define the innermost interior face (refer to the innermost
gray-shaded face of Fig.\ref{fig:octahedron_1}) of the first
representation. The three points of this face lie on the edges of
the triangle defined by points $O_1$, $O_2$ and $O_3$, since circles
$C_1$, $C_2$ and $C_3$ are mutually tangent. Then, at least one of
the angles of the triangle is less than $\pi$, as desired. In the
second representation, the innermost interior face is a circle
(refer to the innermost gray-shaded circle of
Fig.\ref{fig:octahedron_2}) with three distinct points on its
boundary. Trivially, at least one of the arcs defined by those
points corresponds to an angle that is smaller than $\pi$.

\begin{figure}[t]
  \centering
  \begin{minipage}[b]{.48\textwidth}
    \centering
    \subfloat[\label{fig:interior_face1}{}]
    {\includegraphics[width=.6\textwidth]{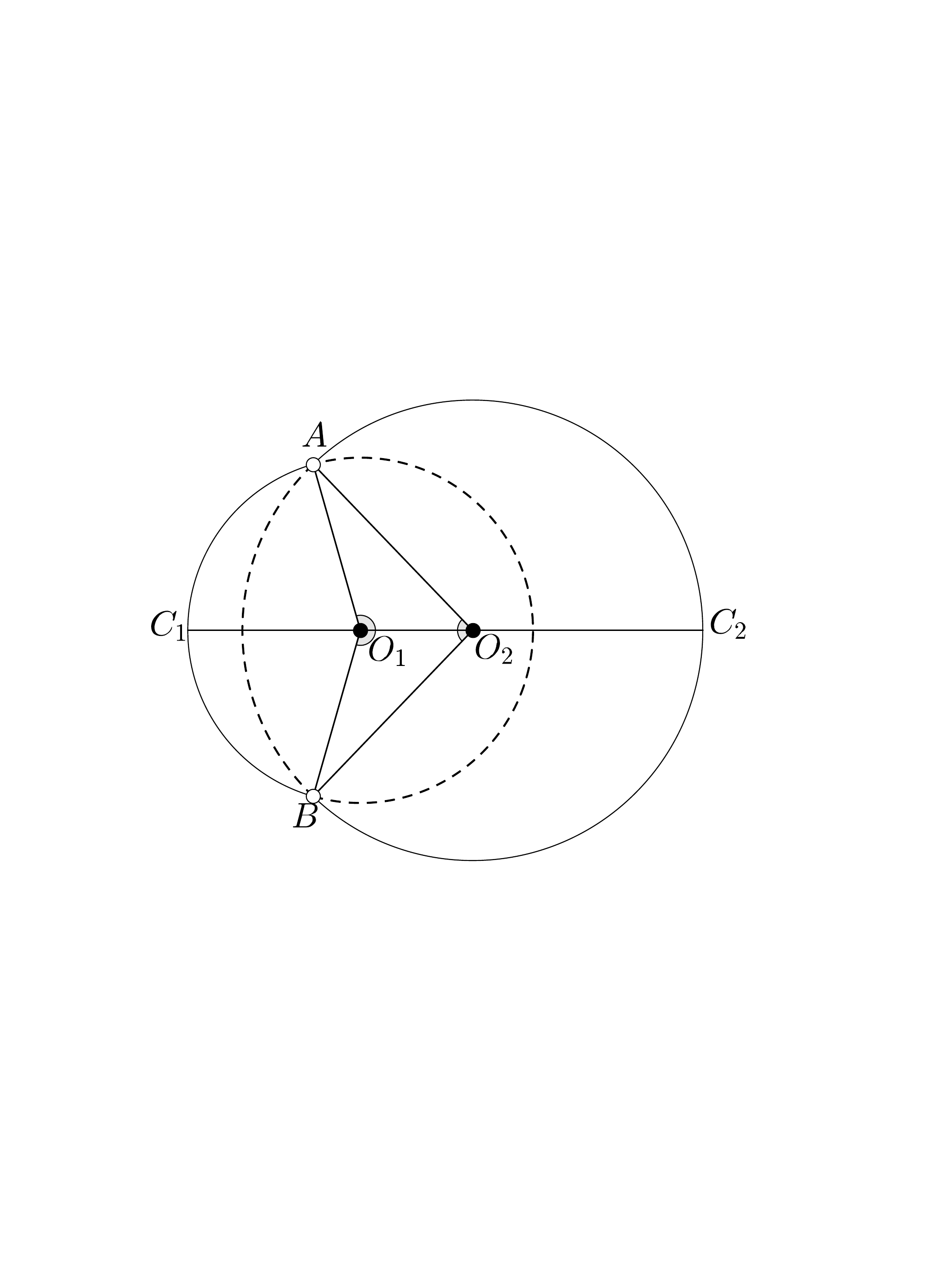}}
  \end{minipage}
  \begin{minipage}[b]{.48\textwidth}
    \centering
    \subfloat[\label{fig:interior_face2}{}]
    {\includegraphics[width=.7\textwidth]{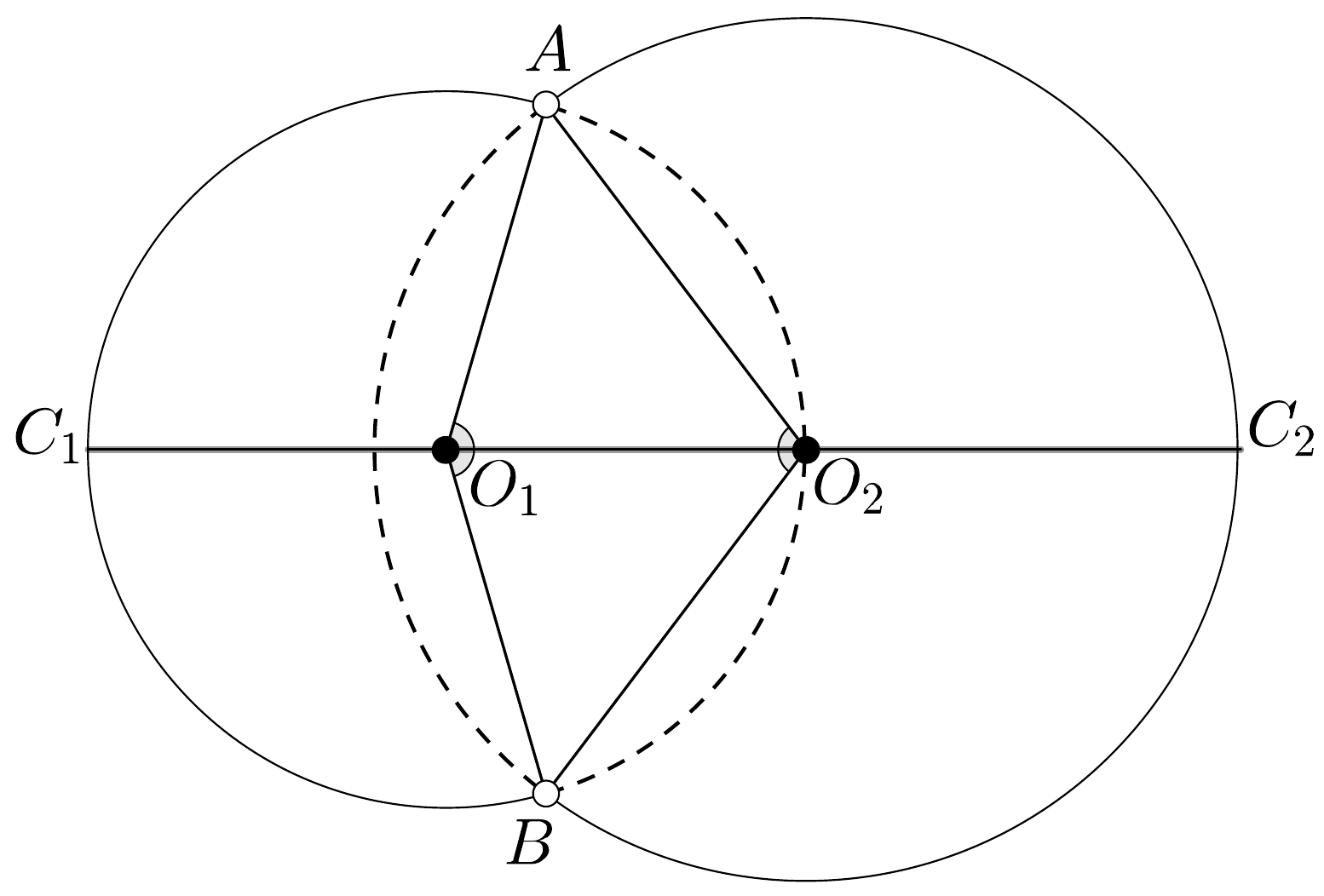}}
  \end{minipage}
  \caption{Configurations considered in Lemma~\ref{lem:face}.}
  \label{fig:interiorface}
\end{figure}

We now turn our attention to the more interesting case where the
realization of the octahedron graph as a system of circles is
implied by three mutually crossing circles (refer to
Fig.\ref{fig:octahedron_3}). First, consider two circles
$C_1(O_1,r_1)$ and $C_2(O_2,r_2)$  intersecting at points $A$ and
$B$ and assume w.l.o.g. that their centers lie along the x-axis,
such that $O_1$ is to the left of $O_2$ (see
Fig.\ref{fig:interiorface}). We are interested in the angles that
correspond to the two arcs of $C_1$ and $C_2$ that ``confine'' the
common points of the two circles (refer to the dashed drawn arcs
incident to $A$ and $B$ in Fig.\ref{fig:interiorface}). It is not
difficult to see that $\widehat{AO_1B}$ and $\widehat{AO_2B}$ cannot
be both greater than $\pi$. Consider now two of the crossing circles
of the realization of the octahedron of Fig.\ref{fig:octahedron_3}.
From the above, it follows that at least one of the two arcs that
``confine'' their common points, corresponds to an angle that is
less than $\pi$. Since the innermost interior face is also confined
by the two arcs, it follows that at least one edge of the innermost
interior face has the desired property.
\end{proof}

\begin{theorem}
There exists a connected 4-regular planar graph that does not admit
a realization as a system of circles. \label{thm:counter}
\end{theorem}
\begin{proof}
Lemma~\ref{lem:face} states that in any realization of the
octahedron graph as a system of circles, at least one edge of the
innermost interior face is drawn as an arc segment of angle less
than $\pi$. Hence, by Lemma~\ref{lem:circle-basic} it follows that
$G_{oct}^{aug}$ does not admit a realization as a system of circles.
\end{proof}

\begin{theorem}
There exists an infinite class of connected 4-regular planar graphs
that do not admit a realization as a system of circles.
\label{thm:class-counter}
\end{theorem}
\begin{proof}
Recall that in order to obtain $G_{oct}^{aug}$, each edge of the
octahedron graph was augmented by two pairs of gadget-subgraphs.
However, Theorem~\ref{thm:counter} trivially holds if more than two
pairs of gadget-subgraphs are attached to each edge of $G_{oct}$,
defining thus an infinite class of connected 4-regular planar graphs
that do not admit a realization as a system of circles. An
alternative (and more interesting) class of such graphs can be
derived by replacing the octahedron graph of the loop-subgraph of
each gadget-subgraph by any 4-regular planar graph, in which one of
the edges on its outerface is replaced by a path of length two and
the additional vertex implied by this procedure is identified by
vertices $w_1$ and/or $w_2$ of the gadget-subgraph (refer to
Fig.\ref{fig:gadget}).
\end{proof}

\section{The case of biconnected 4-regular planar graphs}
\label{sec:biconnected}
In this section, we consider the case of biconnected 4-regular
planar graphs. More precisely, we will prove that there exist
infinitely many biconnected 4-regular planar graphs that do not
admit realizations as system of circles. To do so, we follow a
similar approach as the one presented in
Section~\ref{sec:connected}. Recall that graph $G_{oct}^{aug}$ that
we constructed in Section~\ref{sec:connected} was not biconnected,
since each gadget-subgraph defines two cut-vertices. In fact, all
cut-vertices of $G_{oct}^{aug}$ belong to the gadget-subgraphs. So,
for the case of biconnected 4-regular planar graphs, we will
construct a new gadget-subgraph, referred to as
\emph{bigadget-subgraph}, that does not contain cut-vertices and
simultaneously has the same properties as the corresponding ones of
Section~\ref{sec:connected} (in particular the properties implied by
Lemma~\ref{lem:gadget-realization}).

The bigadget-subgraph is illustrated in
Fig.\ref{fig:gadget-special-2}. Again, it contains exactly two
vertices of degree two, namely $v_1$ and $v_2$, which are its
endpoints. However, its skeleton now consists of seven vertices
(i.e., $v_i$, $w_i$, $w'_i$ and $w$, $i=1,2$). If we remove the
edges of the skeleton except for the edges $(w_i,w'_i)$, $i=1,2$,
the remaining graph again consists of three isolated vertices and
two disjoint biconnected graphs, which we call
\emph{biloop-subgraphs} (refer to the grey-shaded graphs of
Fig.\ref{fig:gadget-special-2}). The properties of the
bigadget-subgraph are again independent of the biloop-subgraphs,
i.e., any simple biconnected planar graph satisfying the following
degree condition can be used instead: Every vertex is of degree 4
except for exactly two vertices on the outerface that are of degree
3. The general situation is shown
in~Fig.\ref{fig:gadget-abstract-2}, where the subgraphs are drawn as
``bi-loops'' at vertices $w_i$ and $w'_i$, $i=1,2$.

\begin{figure}[t]
  \centering
  \begin{minipage}[b]{.48\textwidth}
    \centering
    \subfloat[\label{fig:gadget-special-2}{The bigadget-subgraph.}]
    {\includegraphics[width=\textwidth]{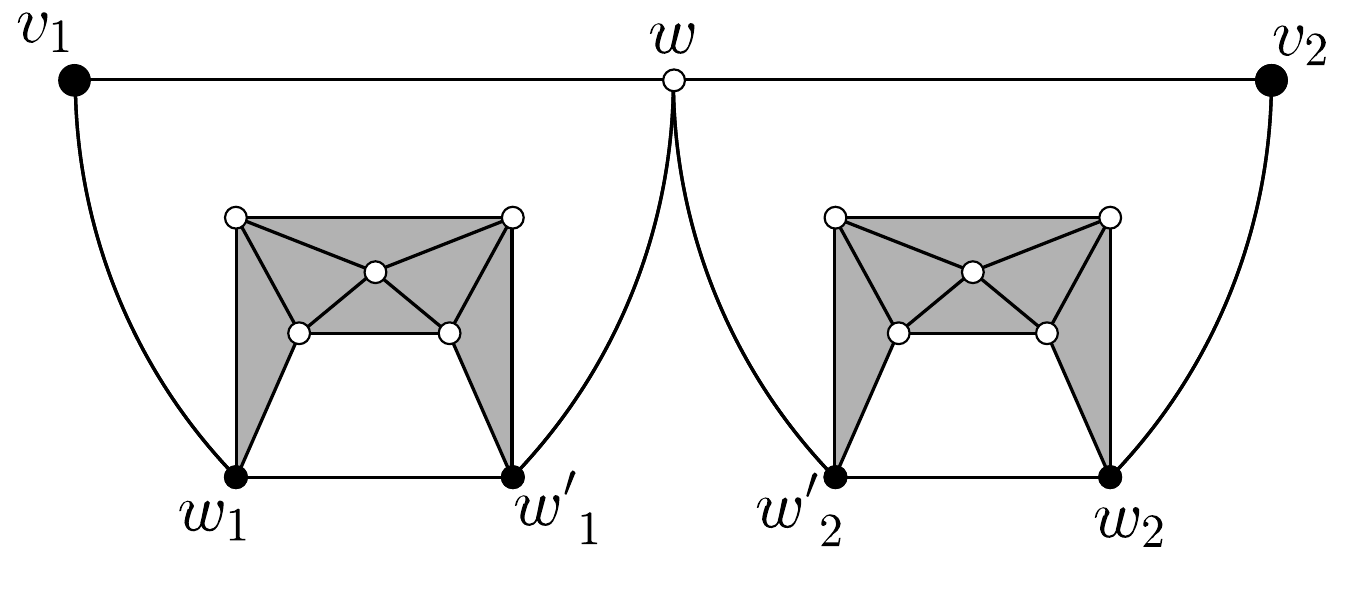}}
  \end{minipage}
  \begin{minipage}[b]{.48\textwidth}
    \centering
    \subfloat[\label{fig:gadget-abstract-2}{Abstraction of the bigadget-subgraph.}]
    {\includegraphics[width=\textwidth]{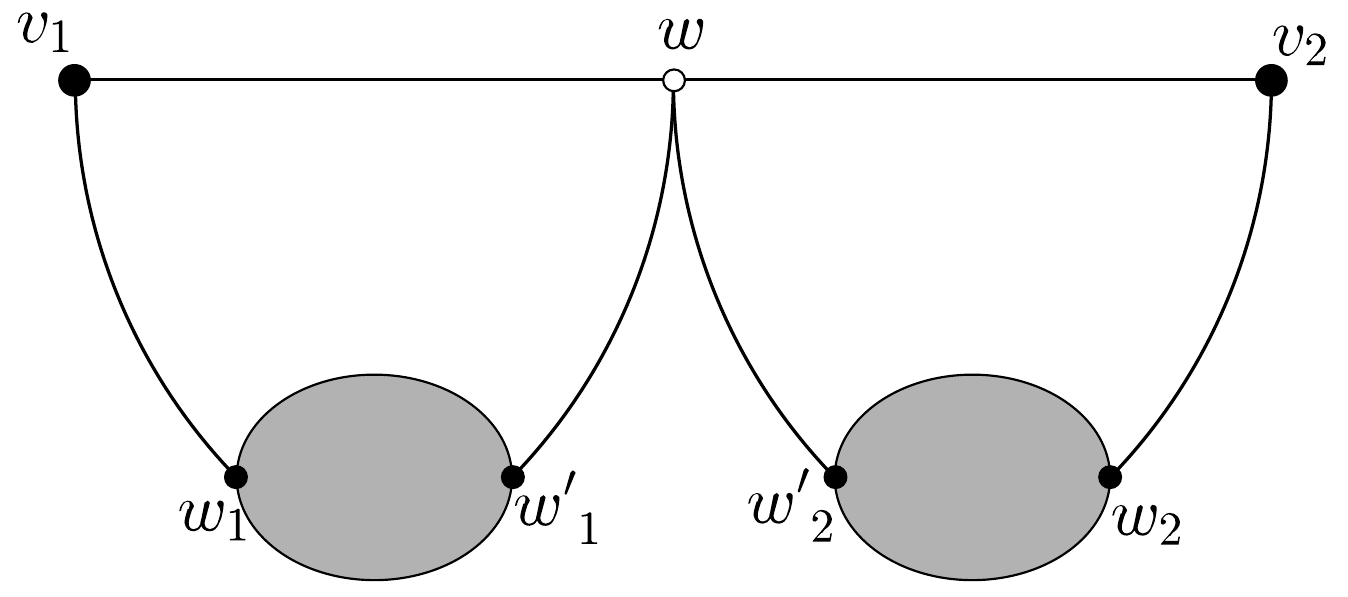}}
  \end{minipage}
  \caption{Illustrations of the new gadget-subgraph.}
  \label{fig:gadget-2}
\end{figure}

Having specified the bigadget-subgraph, we can augment the
octahedron graph similarly to Section~\ref{sec:connected}. This will
result in a biconnected graph, which we denote by $BG_{oct}^{aug}$.
Now, we are in position to prove the analogue of
Lemma~\ref{lem:gadget-realization}.

\begin{lemma}
Let $G$ be a 4-regular planar graph that contains at least one copy
of the bigadget- subgraph. Suppose that there is a realization of
$G$ as a system of circles. Then, the skeleton of each
bigadget-subgraph in this realization consists of two circles $C_1$
and $C_2$ tangent at a point $w$, where circle $C_i$ contains
vertices $\{v_i,w_i,w'_i,w\}$ and the arcs-segments realizing
edges $(v_i,w)$, $(v_i,w_i)$ and $(w'_i,w)$, for $i=1,2$.
\label{lem:gadget-realization-2}
\end{lemma}
\begin{proof}
Suppose that there is a realization of $G$ as a system of circles
and consider a copy of the bigadget-subgraph in this realization.
Note that the edges $(v_i,w_i)$ and $(w'_i,w)$ belong to the same
circle, $i=1,2$ (otherwise, one of the circles would contain vertex
$w_i$ twice and the other circle would contain vertex $w'_i$ twice,
$i=1,2$, which is not possible since every vertex belongs to exactly
two different circles). Let $C_i$ be the circle that contains
$(v_i,w_i)$ and $(w'_i,w)$ and $C'_i$ the circle that contains
$(v_i,w)$, $i=1,2$. We claim that $C_i = C'_i$, $i=1,2$. For the
sake of contradiction, assume that $C_1 \neq C'_1$. Since vertex $w$
is defined by exactly two circles, we have that
$\left\{C_1,C'_1\right\}=\left\{C_2,C'_2\right\}$, which also
implies that $C_2\neq C'_2$. Then, $C_1$ and $C'_1$ have at least
three points in common, namely vertices $v_1$, $v_2$ and $w$, from
which we obtain $C_1=C'_1$; a contradiction.
\end{proof}

\begin{remark}
In the statement of Lemma~\ref{lem:gadget-realization-2} in~\cite{DBLP:journals/jocg/BekosR15}, we erroneously wrote that circle $C_i$ the arc-segments $(v_i,w_i)$, $(w_i,w_i')$ and $(w'_i,w)$,  for $i=1,2$. This was clearly a typo, which we fix in this version.
\end{remark}

In Fig.\ref{fig:gadget-circles-2} two non-equivalent realizations
are shown. Note that these realizations actually depend on the
relative position of the two touching circles $C_1$ and $C_2$ in the
planar embedding of $G$. In particular, in
Fig.\ref{fig:gadget-circles_a-2} the two circles $C_1$ and $C_2$
contain only the biloops-subgraphs in their interior, while in
Fig.\ref{fig:gadget-circles_b-2} $C_1$ is the outerface and
therefore contains the entire graph. This implies that in any
realization of $BG_{oct}^{aug}$ as a system of circles, all
bigadget-subgraphs are drawn as in Fig.\ref{fig:gadget-circles_a-2}
except for at most one bigadget-subgraph if one of its two circles
is the outerface of $BG_{oct}^{aug}$. Hence, we can similarly prove
the analogue of Corollary~\ref{cor:independent-goct}.

\begin{figure}[t]
  \centering
  \begin{minipage}[b]{.53\textwidth}
    \centering
    \subfloat[\label{fig:gadget-circles_a-2}{}]
    {\includegraphics[width=\textwidth]{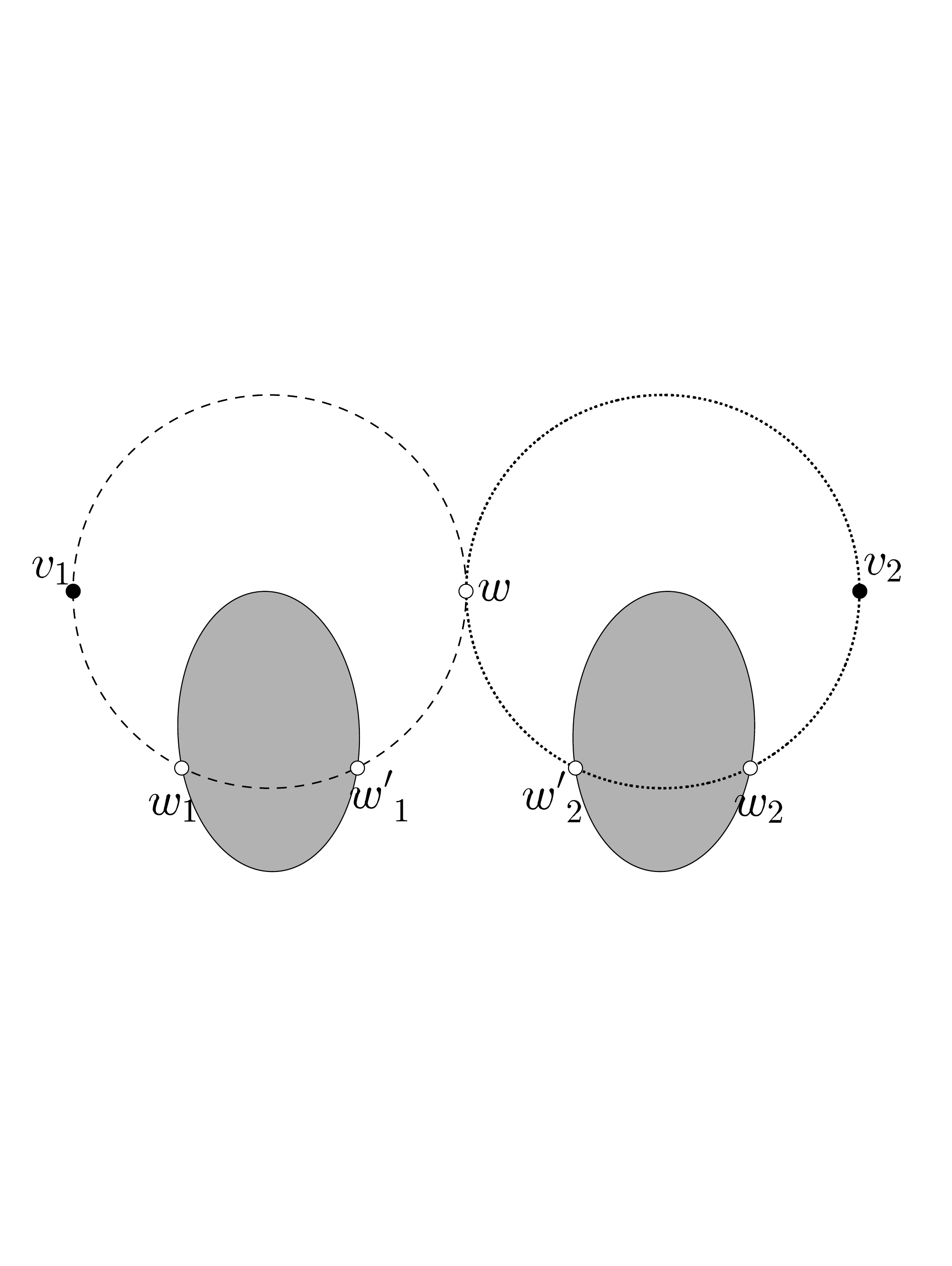}}
  \end{minipage}
  \begin{minipage}[b]{.33\textwidth}
    \centering
    \subfloat[\label{fig:gadget-circles_b-2}{}]
    {\includegraphics[width=\textwidth]{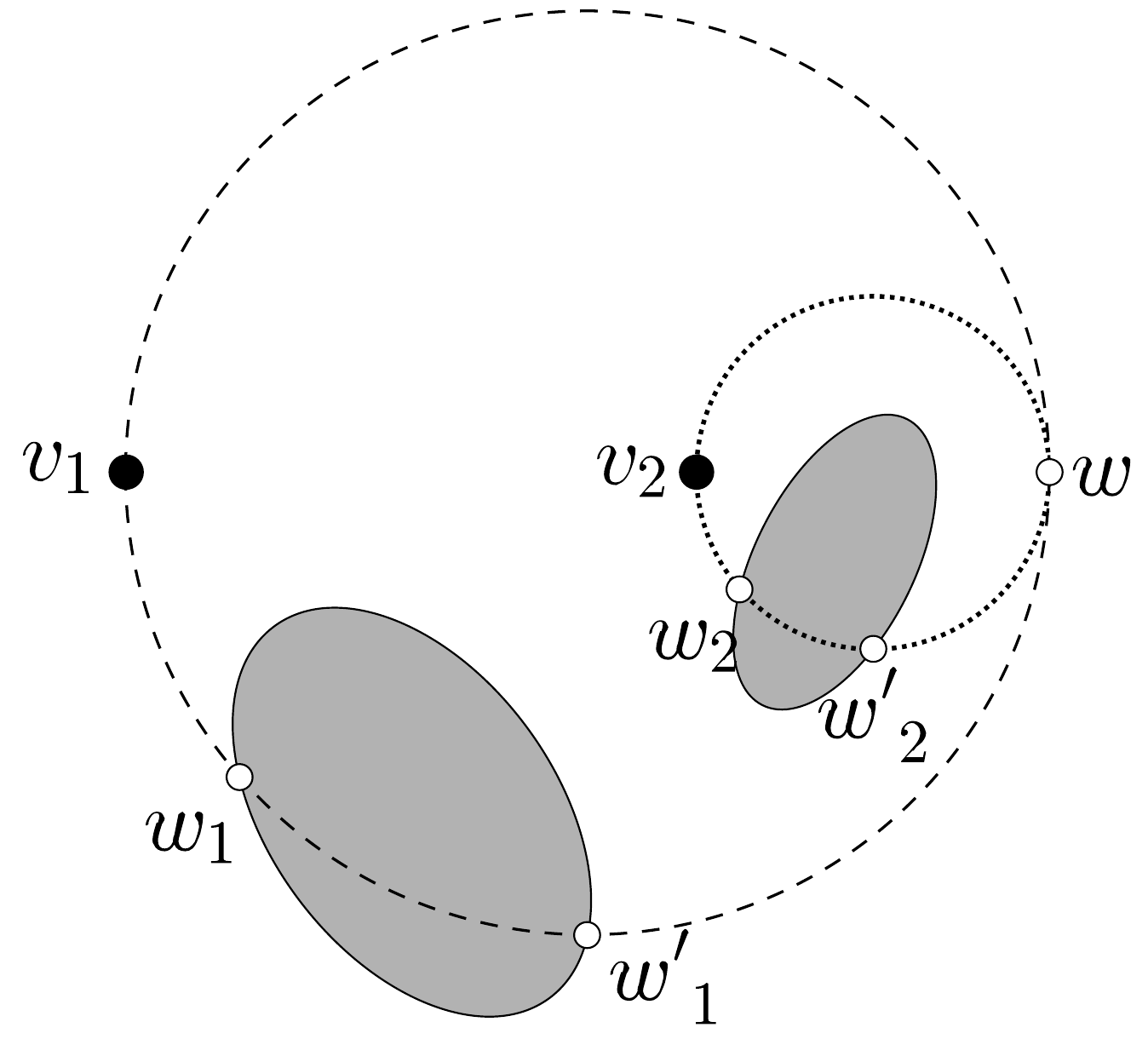}}
  \end{minipage}
  \caption{Non-equivalent realizations of the bigadget-subgraph as system of circles: (a)~the two circles representing the skeleton have disjoint interior, and (b)~one circle lies in the interior of the other.}
  \label{fig:gadget-circles-2}
\end{figure}

\begin{corollary}
In any realization of $BG_{oct}^{aug}$ as a system of circles, the
realizations of each bigadget-subgraph and the realization of
$G^{sub}_{oct}$ are independent. \label{cor:independent-goct-2}
\end{corollary}

Lemma~\ref{lem:gadget-realization-2} and
Corollary~\ref{cor:independent-goct-2} allow us to prove an
analogous of Lemma~\ref{lem:circle-basic} (where bigadget-subgraphs
are used instead of gadget-subgraphs). That, together with
Lemma~\ref{lem:face}, allows us to give an analogue of
Theorem~\ref{thm:counter}.

\begin{theorem}
There exists a biconnected 4-regular planar graph that does not
admit a realization as a system of circles. \label{thm:counter-2}
\end{theorem}

In order to prove that there exist infinitely many biconnected
4-regular planar graphs that do not admit realizations as system of
circles, one can attach more than two pairs of bigadget-subgraphs to
every edge of $G_{oct}$, or replace the biloops-subgraphs by any
simple biconnected planar graph (in which every vertex is of degree
4 except for exactly two vertices on the outerface that are of
degree 3), which leads to the following conclusion.

\begin{theorem}
There exists an infinite class of biconnected 4-regular planar
graphs that do not admit a realization as a system of circles.
\label{thm:class-counter-2}
\end{theorem}

\section{Conclusion - Open Problems}
\label{sec:conclusions}

In this paper, we proved that every 3-connected 4-regular planar
graph admits a realization as a system of touching circles. We also
demonstrated that there exist 4-regular planar graphs which are not
3-connected (i.e., either connected or biconnected) and do not admit
realizations as system of circles. However, our work raises several
open problems. In the following, we name only few of them:

\begin{itemize}
  \item What is the computational complexity of the corresponding
  decision problem, i.e., does a given connected 4-regular planar
  graph admit a realization as a system of circles?
  \item Which is the smallest connected 4-regular planar graph not
  admitting a realization as a system of circles? The ones we manage
  to construct consist of more than 100 vertices.
  \item The octahedron graph admits non-equivalent realizations as system
  of circles, in which the number of circles participating in the
  corresponding realizations also differs. In general, an $n$-vertex
  4-regular planar graph needs at least $(1+\sqrt{1+4n})/2$ and at most
  $2n/3$ circles in order to be realized as a system of circles, as shown
  in Section~\ref{sec:Bounds}. So, what is the range of the number of
  circles needed in order to realize a given (3-connected) 4-regular
  planar graph as a system of circles?
  \item In the context of graph realizations as system of circles,
  it would be interesting to study the class of Eulerian planar graphs.
  Obviously, certain vertices would be defined as the intersection of
  more than two circles.
\end{itemize}

\section*{Acknowledgments}
The authors would like to thank the anonymous reviewers for useful
suggestions and comments that helped in improving the readability of
the paper.

\bibliographystyle{plain}
\bibliography{references}

\end{document}